\documentclass[a4paper,11pt]{article}
\usepackage[dvips]{graphicx}
\usepackage{amssymb}
\usepackage{amsmath}

\usepackage{cite}

\usepackage{sgame}

\usepackage{hyperref}

\usepackage{amsmath}
\usepackage{amscd}
\usepackage{amsthm}
\usepackage{amssymb}
\usepackage{latexsym}
\usepackage{eufrak}
\usepackage{euscript}
\usepackage{epsfig}
\usepackage{graphics}
\usepackage{array}
\usepackage{enumerate}

 \addtocounter{equation}{0}

 \newcommand{\NE}{Nash equilibrium }

\newcommand{\be}{\begin{equation}}
\newcommand{\bel}[1]{\begin{equation}\label{#1}}
\newcommand{\qe}{\end{equation}}
\newcommand{\ee}{\end{equation}}
\newcommand{\eeq}{\end{equation}}
\newcommand{\ba}{\begin{eqnarray}}
\newcommand{\ea}{\end{eqnarray}}
\newcommand{\rf}[1]{(\ref{#1})}
\newcommand{\bi}{\bibitem}

\theoremstyle{theorem}
\newtheorem{theo}{Theorem}
\theoremstyle{theorem}

\theoremstyle{remark}

\theoremstyle{remark}

\theoremstyle{corollary}

\theoremstyle{lemma}
\theoremstyle{definition}
\newtheorem{defi}{Definition}

\usepackage{color}

\begin{document}

\title{Periodic Strategies II: Generalizations and Extensions}
\author{V. K. Oikonomou$^{1}$\thanks{
voiko@physics.auth.gr, v.k.oikonomou1979@gmail.com}, J. Jost$^{1,2}$\thanks{
jost@mis.mpg.de}\\
$^{1}$Max Planck Institute for Mathematics in the Sciences\\
Inselstrasse 22, 04103 Leipzig, Germany\\ $^{2}$Santa Fe
Institute, New Mexico, USA }\maketitle

\begin{abstract}
At a mixed Nash equilibrium, the payoff of a player does not depend on her own action, as long as her opponent sticks to his. In a periodic strategy, a
 concept developed in a
previous paper \cite{jostoikonomou},  in contrast, the own payoff does not depend on the opponent's action. Here, we generalize this to multi-player simultaneous perfect information
strategic form games. 
We show that also in this class of games, there always exists at least one periodic strategy, and we
investigate the mathematical properties of such periodic
strategies. In addition, we demonstrate that
periodic strategies may exist in games with incomplete
information; we shall focus on Bayesian games. Moreover we
discuss the differences between the periodic strategies formalism
and cooperative game theory. In fact, the periodic
strategies are obtained in a purely non-cooperative way, and
periodic strategies are as cooperative as the Nash equilibria are.
Finally, we incorporate the periodic strategies in an epistemic
game theory framework, and discuss several features of this
approach.
\end{abstract}

\section{Introduction}

John Nash \cite{nash} showed that every strategic form game
possesses at least one Nash equilibrium (for an alternative proof,
that avoids the use of Brouwer's fixed point theorem and only
needs simple topological facts about bifurcations, see
\cite{jbow}). Here, it is assumed that players act rationally in
the sense that they try to maximize their payoffs, and this
rationality of all players is common knowledge, as are the
possible actions and payoffs of each player. The Nash equilibrium
then is consistent in the sense that when everybody plays it, no
single player could gain an advantage from a unilateral deviation.
Such a Nash equilibrium can be pure, that is, each player plays
some definite strategy, or mixed, where some players choose among
their actions with certain probabilities. For instance, the
matching pennies game has only one Nash equilibrium, and this is
mixed, as each player plays either strategy randomly with
probability $1/2$. Such a mixed \NE has a curious property. To see this, for simplicity, we consider a game
with two players $i=A,B$ who have two possible actions $1,2$ each.
When $A$ and $B$ play actions $\alpha$ and $\beta$ with respective
probabilities $p_\alpha$ and $q_\beta$ (with $p_1+p_2=1=q_1+q_2$),
then the (expected) utility of $A$ is (in obvious notation)
\begin{equation}
  \label{int1}
\mathcal{U}_A=\sum_{\alpha, \beta} \mathcal{U}_A(\alpha,\beta)p_\alpha q_\beta.
\end{equation}
When now $A$ wants to maximize her payoff, she adjusts her probabilities $p_1$ and applies calculus to get as a first order necessary condition at a mixed value $0<p_1<1$
\begin{equation}
  \label{int2}
 0= \sum_{\beta} \mathcal{U}_A(1,\beta) q_\beta -\sum_{\beta} \mathcal{U}_A(2,\beta) q_\beta .
\end{equation}
This then is a condition about the probabilities $q_\beta$ of her
opponent which is independent of her own probabilities $p_\alpha$.
That is, when the opponent plays according to those values, it is
irrelevant for $A$ what she plays. She will always get the same
payoff. Thus, at a mixed Nash equilibrium, when every player has a
mixed strategy, no single player can change her outcome by
changing her strategy, as long as all others stick to their
probabilities.

Of course, this is well known. The phenomenon is simply a
consequence of the fact that the utility $\mathcal{U}$ depends
{\it linearly} on the probabilities of the individual players.
Therefore, taking the derivative w.r.t. them makes the resulting
condition  independent of them.

In \cite{jostoikonomou}, we have investigated what happens when
$A$ seeks a critical point of \eqref{int1} not with respect to her
own probability $p_!$, but with respect to the opponent
probability $q_!$. We then get the condition
\begin{equation}
  \label{int3}
 0= \sum_{\alpha} \mathcal{U}_A(\alpha,1) p_\alpha -\sum_{\alpha} \mathcal{U}_A(\alpha,2) p_\alpha .
\end{equation}
This is now independent of the opponent's probabilities $q_\beta$. That is, when $A$ plays according to \eqref{int3}, her payoff is unaffected by the choice of strategy of her opponent.

Let us consider a simple example  where the  payoff table is given by,
\begin{table}[h]
\centering
  \begin{tabular}{| l |l |l | }
    \hline
      $ $ &  $1$ & $2$ \\ \hline
  $1$ & 2,1 & 0,0   \\ \hline
    $2$ & 0,0 & 1,1\\
    \hline
 \end{tabular}
\end{table}\\
with $A$ being the row player and $B$ the column player. \eqref{int2} for $A$ yields $q_1=1/3, q_2=2/3$, and the analogous computation for $B$ gives $p_1=1/2=p_2$. The expected payoffs for at this mixed Nash equilibrium are
\bel{int4}
\overline{\mathcal{U}}_A= 2 \frac{1}{2} \frac{1}{3} + 1 \frac{1}{2}\frac{2}{3}=\frac{2}{3},\quad \overline{\mathcal{U}}_B= 1 \frac{1}{2} \frac{1}{3} + 1 \frac{1}{2}\frac{2}{3}=\frac{1}{2}.
\qe
In contrast, \rf{int3}, when applied for both players, yields $p_1=\frac{1}{3}, q_1=\frac{1}{2}$. The expected utilities remain the same. As investigated in \cite{jostoikonomou}, the latter property does not always hold, that is, the payoffs at a mixed Nash and at an equilibrium computed according to \rf{int3} need not always be the same, and depending on the game, either of them could be larger than the other. But an equilibrium according to \rf{int3}, called {\it periodic} for reasons to be discussed in a moment, exists in the same generality as a mixed \NE and to show this is the main purpose  of this paper.

In order to explore the consistency of such an equilibrium, it is
useful to recall the concept of rationalizability of Bernheim
\cite{bernheim} and Pearce \cite{pearce}. Here, a sequence of
alternating strategy choices of $A$ and $B$ is called
rationalizable if each of them is a best response to the previous
strategy of the opponent. A \NE is rationalizable, but in general,
there are other sequences of rationalizability strategies. For
instance, in the matching pennies game, there is a sequence where
each player alternates between her/his two options. In such a
sequence, each strategy periodically repeats itself. A similar
phenomenon exists also for our type of equilibrium, and this is
the reason why it was called {\it periodic}. The interpretation is
somewhat curious, however. Since a player can of course not
directly choose the opponent's probabilities to maximize her
payoff, which was the assumption underlying \rf{int3}, the logic
has to become somewhat different. While $A$ cannot choose $q_1$,
her opponent can choose his $q_1$ so as to maximize $A$'s payoff,
and conversely, $A$ could then choose her $p_1$ to maximize $B$'s
payoff. Again, this can be done as an iterative response as in the
rationalizability paradigm, and when both players act that way,
this is perfectly self-consistent. That is, when everybody
believes  that everybody else operates in that way and acts
accordingly, a periodic cycle exists that confirms everybody's
belief.

 In this paper,  we first generalize
the theorems related to periodic strategies to simultaneous
multi-player perfect information strategic form games. Several examples will illustrate the new features brought by the
presence of three or more players.

Periodic strategies for non-trivial perfect information
simultaneous strategic form games are related (or in some cases
are identical) to all the existing rationalizable strategies
\cite{rat1,rat2,rat3,rat4,rat5,rat6,rat7,rat8,rat9,rat10,rat11,rat12,rat13,rat14,rat15,rat16,rat17},
as demonstrated in \cite{jostoikonomou}. We shall then turn to the
question whether such  periodic strategies also  exist in
strategic form games with incomplete information. We shall mainly
focus on Bayesian games
\cite{rat3,harsanyi1,harsanyi2,harsanyi3,tan88,new2,new3,new4,new5,new6,new7,new8},
in which case the presence of rationalizable strategies (ex ante
and interim) suggests that periodic strategies should also exist
in this type of games. In fact, games where the players are
uncertain about the setting and only know that certain scenarios
occur with certain probabilities can sometimes be modelled as a
games with an additional nature player who chooses among the
scenarios with those probabilities. This will also be useful for
our reasoning. Also with regard to incomplete information and
cooperative game theory, in Ref. \cite{ref1} an interesting
approach was used in order to study reentrant phase transitions
and defensive alliances in social dilemmas with informed
strategies. Furthermore a review on co-evolutionary games was
provided in Ref. \cite{ref2}.

One of the most important features of the periodic strategy
algorithm is that periodic strategies do apply to non-cooperative
game theory \cite{tirole,ozzy}. As explained, however, by construction, the
periodic strategies are based on maximizing the payoffs of a game
for a player by using the probability distribution of the
opponents. We should point out, however, that this is different from the setting of cooperative game theory, which is about coalitions and distributions of payoffs inside them. In contrast,  the
procedure for obtaining the periodic strategies involves
maximization of a player's own utility function, without any
cooperation with the opponent, or any apparent agreement. Finally,
we shall attempt to incorporate  the periodicity concept
into an epistemic game theory \cite{tan88,pereabook} theoretical
framework.

The outline of the paper is as follows: In section 2, we
generalize the periodic solution concept to
multi-player finite, perfect information simultaneous strategic
form games. In section 3 we study the periodic solution concept
for games with incomplete information, quantified in terms of
Bayesian games. The non-cooperativity argument on which
periodicity is based is discussed in section 4, while in section
5, we incorporate the periodicity concept in a very simple
epistemic game theory framework by connecting types to the
periodicity number, without getting into much details however. The
conclusions along with future perspectives of the periodicity
concept follow at the end of the paper.

\section{Generalization of the Periodicity Concept to Multi-player Games -- The Perfect Information Case}

In this section we generalize the concept of periodicity to
multi-player strategic form games with perfect information. We
start  with a concrete example. Consider a three
player game with 
\begin{itemize}
\item The set of players: $I={A,B,C}$

\item Their strategy spaces $\mathcal{M}(A)$, $\mathcal{M}(B)$,  $\mathcal{M}(C)$ and the total strategy
space $\bar{G}=\mathcal{M}(A)\times \mathcal{M}(B)\times  \mathcal{M}(C)$

\item The payoff functions
$\mathcal{U}_i(\bar{G}):\bar{G}\rightarrow \Re$, $i={A,B,C}$

\end{itemize}
We define six continuous maps  between the strategy spaces $\mathcal{M}(i)$ and
$\mathcal{M}(j)$, 
\begin{equation}\label{phi1}
\varphi_{ij}:\mathcal{M}(i)\rightarrow \mathcal{M}(j)
\end{equation}
We usually write $\varphi_{ij}\varphi_{km}$ for 
$\varphi_{ij}\circ\varphi_{km}$.
The maps $\varphi_{ij}$ and $\varphi_{ji}$, act in such a way that when  we start with an action $x_k$ of player $i$, the following inequality holds:
\begin{align}\label{inequality}
&\mathcal{U}_i(x_k,\varphi_{ij}(x_k),\varphi_{im}(x_k))> \mathcal{U}_i(x_k,x,y) {\,}{\,}{\,}{\,}{\,}{\,}\forall {\,}(x,y){\,}\in {\,}\mathcal{M}(j)\cup \mathcal{M}(m)\backslash\{\varphi_{ij}(x_k),\varphi_{im}(x_k)\}
\end{align}
In the  example of GAME 1  in Fig.
\ref{3PLAYERGAMEa1}, each player has two actions available.
\begin{figure}[h!]
\centering
\includegraphics[width=25pc]{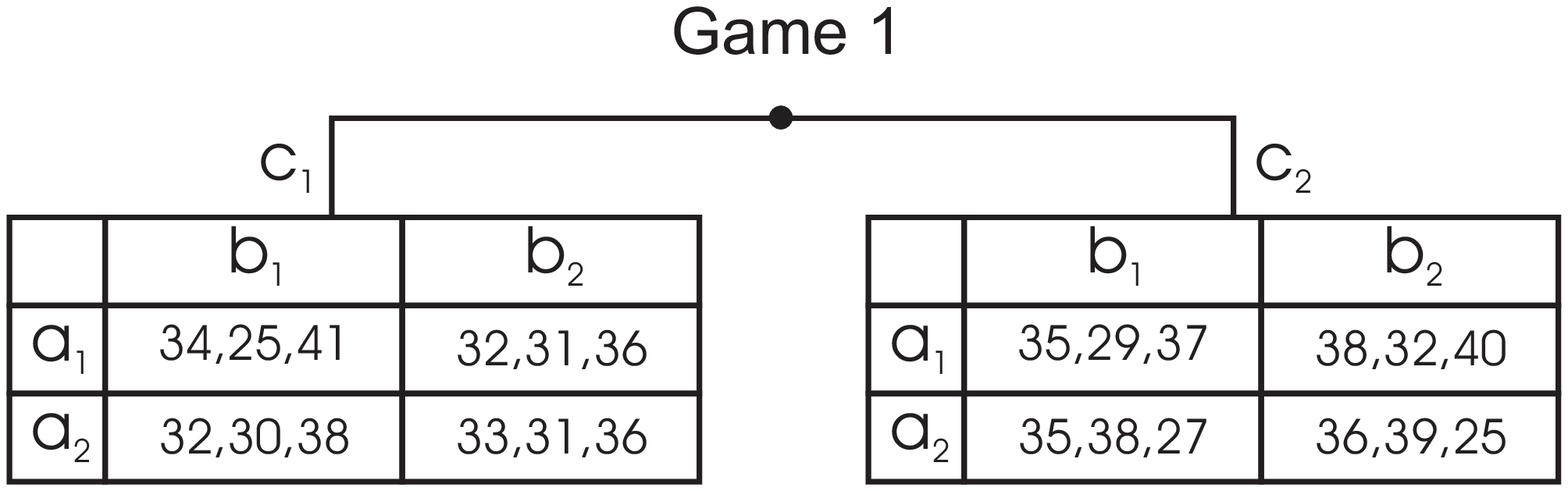}
\caption{A 3-Player Game payoff matrix. The game is a simultaneous
action game of three players A, B and C. The actions of the
players A, B and C are denoted as $(a_1,a_2)$, $(b_1,b_2)$ and
$(c_1,c_2)$ respectively. }\label{3PLAYERGAMEa1}
\end{figure}

In Fig. \ref{3PLAYERGAMEanew} we can see the periodicity chains
for the action $a_1$ of player A,  recalling  the periodicity concept we gave in the 2-player game
case in Ref. \cite{jostoikonomou}. Let us give a verbal description
of the periodicity diagram. The letters
$ABC$ on the arrows indicate the  player whose action is considered.

Player A will play $a_1$ if player B plays $b_2$ and player C
plays $c_2$ simultaneously. In the map notation, this becomes 
$\varphi_{AB}(a_1),\varphi_{AC}(a_1)$, as indicated in the figure.
By following the B arrow, B will
play $b_2$ if player A plays $a_2$ and player C plays $c_2$.
Following C in node ''1'', C would play $c_2$ if player B plays
$b_2$ and A plays $a_1$ (we have reached a periodic cycle at this point
but we continue in order to show the new structures). Back in node 2,
following the C arrow, C will play $c_2$ if B plays $b_2$ and A
plays $a_1$. Back in node 2 following the arrow A, A will play
$a_2$ if B plays $b_2$ and C plays $c_2$. Accordingly, in node 3,
following arrow B, B will play $b_2$ if A plays $a_2$ and C plays
$c_2$ (we have reached a set stable cycle of $a_2$ as we will see) and
so on.
\begin{figure}[h!]
\centering
\includegraphics[width=25pc]{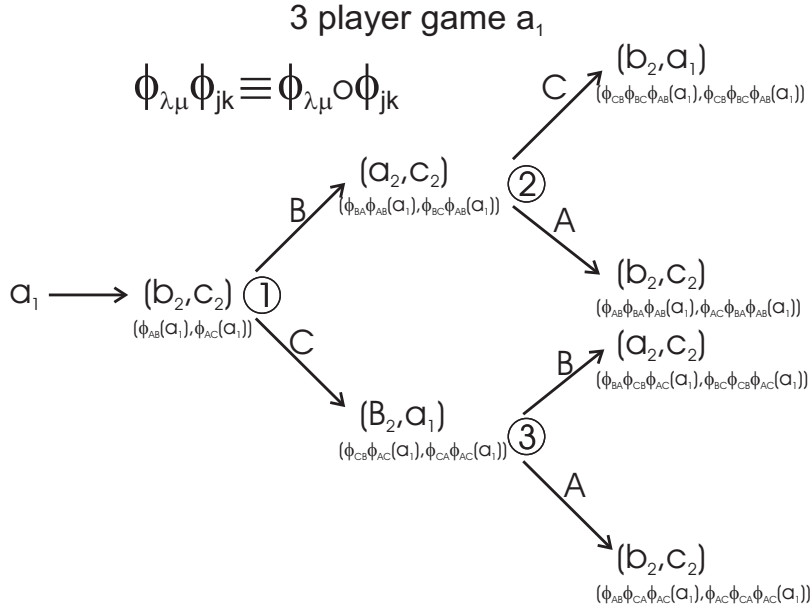}
\caption{The periodicity of the strategy $a_1$ for the 3-player
game of Fig. \ref{3PLAYERGAMEa1}. In the figure it is shown how
the periodicity concept is realized for the strategy $a_1$ of
player A in detail. The graph returns to the original strategy
$a_1$.}\label{3PLAYERGAMEanew}
\end{figure}
Back at node 3, following A, A will play $a_1$, if B plays $b_2$
and $C$ plays $c_2$ and so on. Thus, we have the periodic cycles
\begin{align}\label{periodicitytypes}
& \varphi_{CA}\varphi_{AC}(a_1)=a_1 \\ \notag
& \varphi_{CA}\varphi_{BC}\varphi_{AB}(a_1)=a_1 \\ \notag
\end{align}
The most striking new feature of the multi-player game case
is the fact that in the periodicity algorithm, the
utility functions appear in a rather different order as we shall
see. Let us take the first  type, 
$\varphi_{CA}\varphi_{AC}(a_1)=a_1$. The periodic algorithm in
terms of the utility functions is
\begin{align}\label{qper1}
&\mathcal{U}_A (a_1,\varphi_{AB}(a_1),\varphi_{AC}(a_1))> \mathcal{U}_A(a_1,x,y) {\,}{\,}{\,}{\,}{\,}{\,}
 \forall {\,}(x,y){\,}\in {\,}\mathcal{M}(B)\cup \mathcal{M}(C)\backslash\{\varphi_{AB}(a_1),\varphi_{AC}(a_1)\}
\\  & \notag
\mathcal{U}_C(\varphi_{AC}(a_1),\varphi_{CA}\varphi_{AC}(a_1),\varphi_{CB}\varphi_{AC}(a_1))> \mathcal{U}_C(\varphi_{AC}(a_1),x,y) {\,}{\,}{\,}{\,}{\,}{\,} \forall {\,}(x,y){\,}\\ \notag & \in {\,}\mathcal{M}(A)\cup \mathcal{M}(B)\backslash\{\varphi_{CA}\varphi_{AC}(a_1),\varphi_{CB}\varphi_{AC}(a_1)\}
\\ \notag & \mathcal{U}_A(\varphi_{CA}\varphi_{AC}(a_1),\varphi_{AB}\varphi_{CA}\varphi_{AC}(a_1),\varphi_{AC}\varphi_{CA}\varphi_{AC}(a_1))> \mathcal{U}_A(\varphi_{CA}\varphi_{AC}(a_1),x,y) {\,}{\,}{\,}{\,}{\,}{\,}
\\ \notag & \forall {\,}(x,y){\,}\in {\,}\mathcal{M}(C)\cup \mathcal{M}(B)\backslash\{\varphi_{AB}\varphi_{CA}\varphi_{AC}(a_1),\varphi_{AC}\varphi_{CA}\varphi_{AC}(a_1)\}
\end{align}
For the other type, 
$\varphi_{CA}\varphi_{BC}\varphi_{AB}(a_1)=a_1$, the periodic algorithm becomes
\begin{align}\label{qper1}
& \mathcal{U}_A(a_1,\varphi_{AB}(a_1),\varphi_{AC}(a_1))> \mathcal{U}_A(a_1,x,y) {\,}{\,}{\,}{\,}{\,}{\,}
 \forall {\,}(x,y){\,}\in {\,}\mathcal{M}(B)\cup \mathcal{M}(C)\backslash\{\varphi_{AB}(a_1),\varphi_{AC}(a_1)\}
\\ \notag &
\mathcal{U}_B(\varphi_{AB}(a_1),\varphi_{BA}\varphi_{AB}(a_1),\varphi_{BC}\varphi_{AB}(a_1))> \mathcal{U}_C(\varphi_{AB}(a_1),x,y) {\,}{\,}{\,}{\,}{\,}{\,} \forall {\,}(x,y){\,}\\ \notag & \in {\,}\mathcal{M}(A)\cup \mathcal{M}(B)\backslash\{\varphi_{BA}\varphi_{AB}(a_1),\varphi_{BC}\varphi_{AB}(a_1)\}
 \\ \notag & \mathcal{U}_C(\varphi_{BC}\varphi_{AB}(a_1),\varphi_{CA}\varphi_{BC}\varphi_{AB}(a_1),\varphi_{CB}\varphi_{BC}\varphi_{AB}(a_1))> \mathcal{U}_C(\varphi_{BC}\varphi_{AB}(a_1),x,y) {\,}{\,}{\,}{\,}{\,}{\,}\\ \notag & \forall {\,}(x,y){\,}\in {\,}\mathcal{M}(C)\cup \mathcal{M}(B)\backslash\{\varphi_{CA}\varphi_{BC}\varphi_{AB}(a_1),\varphi_{CB}\varphi_{BC}\varphi_{AB}(a_1)\}
\end{align}
In terms of utility functions, this looks like
\begin{align}\label{g2perc}
& U_A{\,}{\,}{\,}\xrightarrow{P}{\,}{\,}{\,} U_C{\,}{\,}{\,}\xrightarrow{P}{\,}{\,}{\,} U_A \\ \notag &
U_A{\,}{\,}{\,}\xrightarrow{P}{\,}{\,}{\,} U_C{\,}{\,}{\,}\xrightarrow{P}{\,}{\,}{\,} U_A \xrightarrow{P}{\,}{\,}{\,} U_C{\,}{\,}{\,}\xrightarrow{P}{\,}{\,}{\,} U_A  \\ \notag &
U_A{\,}{\,}{\,}\xrightarrow{P}{\,}{\,}{\,} U_B{\,}{\,}{\,}\xrightarrow{P}{\,}{\,}{\,} U_B{\,}{\,}{\,}\xrightarrow{P}{\,}{\,}{\,} U_B \\ \notag &
U_A{\,}{\,}{\,}\xrightarrow{P}{\,}{\,}{\,} U_B{\,}{\,}{\,}\xrightarrow{P}{\,}{\,}{\,} U_A{\,}{\,}{\,}\xrightarrow{P}{\,}{\,}{\,} U_C \\ \notag &
U_A{\,}{\,}{\,}\xrightarrow{P}{\,}{\,}{\,} U_C{\,}{\,}{\,}\xrightarrow{P}{\,}{\,}{\,} U_B {\,}{\,}{\,}\xrightarrow{P}{\,}{\,}{\,} U_C\\ \notag &
\end{align}
If we include all the periodic
points we found in the graph, we have the following new types of
periodicity (some of which belong to set stable cycles):
\begin{align}\label{periodicitytypes}
& \varphi_{CA}\varphi_{AC}(a_1)=a_1 \\ \notag
& \varphi_{BC}\varphi_{AB}(a_1)=a_1 \\ \notag
& \varphi_{BC}\varphi_{CB}\varphi_{AC}(a_1)=a_1 \\ \notag
& \varphi_{AC}\varphi_{BA}\varphi_{AB}(a_1)=a_1 \\ \notag
& \varphi_{BC}\varphi_{AB}\varphi_{CA}\varphi_{AC}(a_1)=a_1 \\ \notag
& \varphi_{BC}\varphi_{CB}\varphi_{AC}(a_1)=a_1 \\ \notag
\end{align}
The periodicity corresponding to the $a_2$ action is shown in Fig. \ref{3Playergamea2}.
\begin{figure}[h!]
\centering
\includegraphics[width=25pc]{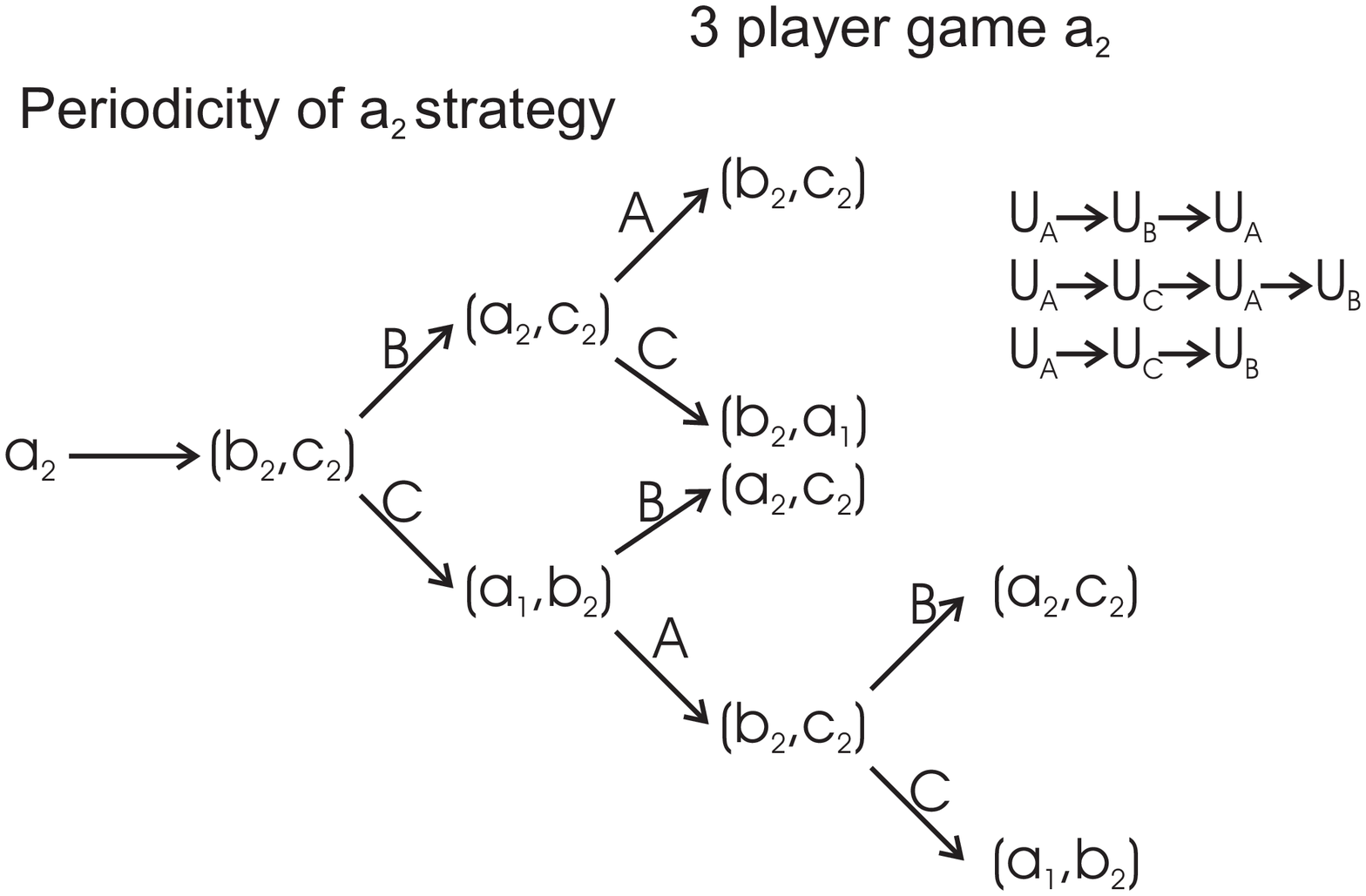}
\caption{Periodicity of strategy $a_2$ for the 3-player game of
Fig. \ref{3PLAYERGAMEa1}.}\label{3Playergamea2}
\end{figure}
In general, there can be 
various types of periodicity, with their number, type and form not
directly depending on
the numbers of players and  actions. As the number
of players increases, depending on the payoffs, the complexity of
the periodic strategies significantly increases. But as will
become obvious, the complexity of the algorithm depends strongly
on the payoffs. Now we generalize this type of games and we
proceed to a 4-player game with each player having again two available
actions, as shown in Fig. \ref{4playergame}. In Fig.
\ref{4playera1} and
\ref{4playera2}, we see the periodic structure for the actions $a_1$ and $a_2$, resp.
\begin{figure}[h!]
\centering
\includegraphics[width=35pc]{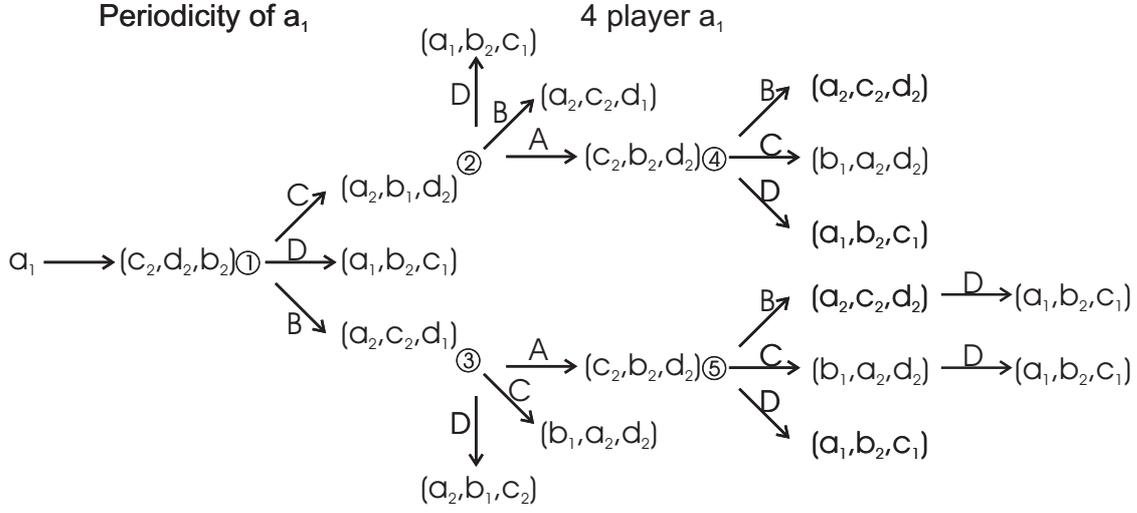}
\caption{Periodicity of strategy $a_1$ for the 4-player game of
Fig. \ref{4playergame}.}\label{4playera1}
\end{figure}
 We now look at the periodic strategies for $a_1$ in Fig. \ref{4playera1}. These are the following:
\begin{align}\label{periodicitytypes4P}
& \varphi_{DA}\varphi_{AD}(a_1)=a_1 \\ \notag
& \varphi_{DA}\varphi_{CD}\varphi_{AC}(a_1)(a_1)=a_1 \\ \notag
& \varphi_{DA}\varphi_{AD}\varphi_{CA}(a_1)\varphi_{AC}(a_1)=a_1 \\ \notag
& \varphi_{DA}\varphi_{BD}\varphi_{AB}(a_1)\varphi_{BA}(a_1)\varphi_{AB}(a_1)=a_1 \\ \notag
& \varphi_{DA}\varphi_{AD}\varphi_{BA}\varphi_{AB}(a_1)=a_1 \\ \notag
\end{align}
In terms of utility functions, this looks as follows, with Type $j$ referring to line $j$ in \eqref{periodicitytypes4P}.
\begin{figure}[h!]
\centering
\includegraphics[width=25pc]{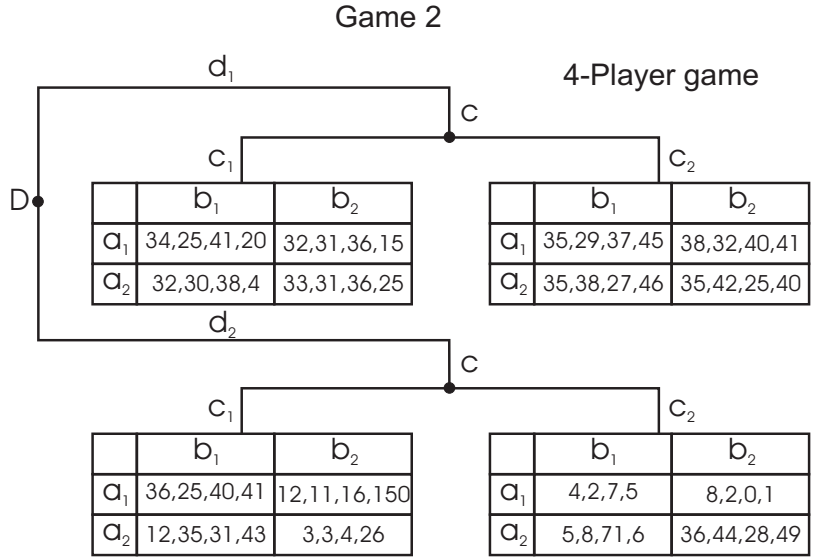}
\caption{A 4-player game payoff matrix. The game is a simultaneous
action game of four players A, B, C and D.}\label{4playergame}
\end{figure}

\textbf{Type 1}

\begin{align}\label{qper1}
& \mathcal{U}_A(a_1,\varphi_{AD}(a_1),\varphi_{AC}(a_1),\varphi_{AB}(a_1))> \mathcal{U}_A(a_1,x,y,z) {\,}{\,}{\,}{\,}{\,}{\,}
 \\ \notag & \forall {\,}(x,y,z){\,}\in {\,}\mathcal{M}(B)\cup \mathcal{M}(C)\cup \mathcal{M}(D)\backslash\{,\varphi_{AD}(a_1),\varphi_{AC}(a_1),\varphi_{AB}(a_1)\}
\\  &
\mathcal{U}_D(\varphi_{AD}(a_1),\varphi_{DA}\varphi_{AD}(a_1),\varphi_{DC}\varphi_{AD}(a_1),\varphi_{DB}\varphi_{AD}(a_1))> \mathcal{U}_D(\varphi_{AD}(a_1),x,y,z) {\,}{\,}{\,}{\,}{\,}{\,} \\ \notag &\forall {\,}(x,y,z){\,} \in {\,}\mathcal{M}(A)\cup \mathcal{M}(B)\cup \mathcal{M}(C)\backslash\{\varphi_{BA}\varphi_{AB}(a_1),\varphi_{BC}\varphi_{AB}(a_1)\}
\end{align}

\textbf{Type 2}

\begin{align}\label{qper1}
& \mathcal{U}_A(a_1,\varphi_{AC}(a_1),\varphi_{AD}(a_1),\varphi_{AB}(a_1))> \mathcal{U}_A(a_1,x,y,z) {\,}{\,}{\,}{\,}{\,}{\,}
 \\ \notag & \forall {\,}(x,y,z){\,}\in {\,}\mathcal{M}(B)\cup \mathcal{M}(C)\cup \mathcal{M}(D)\backslash\{,\varphi_{AD}(a_1),\varphi_{AC}(a_1),\varphi_{AB}(a_1)\}
\\ \notag &
\mathcal{U}_C(\varphi_{AC}(a_1),\varphi_{CA}\varphi_{AC}(a_1),\varphi_{CB}\varphi_{AC}(a_1),\varphi_{CD}\varphi_{AC}(a_1))> \mathcal{U}_C(\varphi_{AC}(a_1),x,y,z) {\,}{\,}{\,}{\,}{\,}{\,} \forall {\,}(x,y,z){\,}\\ \notag & \in {\,}\mathcal{M}(A)\cup \mathcal{M}(B)\cup \mathcal{M}(D)\backslash\{,\varphi_{CA}\varphi_{AC}(a_1),\varphi_{CB}\varphi_{AC}(a_1),\varphi_{CD}\varphi_{AC}(a_1)\}
\\ \notag & \mathcal{U}_D(\varphi_{CD}\varphi_{AC}(a_1),\varphi_{DA}\varphi_{CD}\varphi_{AC}(a_1),\varphi_{DB}\varphi_{CD}\varphi_{AC}(a_1),\varphi_{DC}\varphi_{CD}\varphi_{AC}(a_1))>
\\ \notag &\mathcal{U}_C(\varphi_{CD}\varphi_{AC}(a_1),x,y,z) {\,}{\,}{\,}{\,}{\,}{\,}
\\ \notag &\forall {\,}(x,y,z){\,}\in {\,}\mathcal{M}(C)\cup \mathcal{M}(B)\cup \mathcal{M}(A)\backslash\{\varphi_{DA}\varphi_{CD}\varphi_{AC}(a_1),\varphi_{DB}\varphi_{CD}\varphi_{AC}(a_1),\varphi_{DC}\varphi_{CD}\varphi_{AC}(a_1)\}
\end{align}

\textbf{Type 3}

\begin{align}\label{qper1}
& \mathcal{U}_A(a_1,\varphi_{AC}(a_1),\varphi_{AD}(a_1),\varphi_{AB}(a_1))> \mathcal{U}_A(a_1,x,y,z) {\,}{\,}{\,}{\,}{\,}{\,}
 \\ \notag & \forall  (x,y,z) \\ \notag &{\,}\in {\,}\mathcal{M}(B)\cup \mathcal{M}(C)\cup \mathcal{M}(D)\backslash\{,\varphi_{AD}(a_1),\varphi_{AC}(a_1),\varphi_{AB}(a_1)\}
\\ \notag &
\mathcal{U}_C(\varphi_{AC}(a_1),\varphi_{CA}\varphi_{AC}(a_1),\varphi_{CB}\varphi_{AC}(a_1),\varphi_{CD}\varphi_{AC}(a_1))>
\mathcal{U}_C(\varphi_{AC}(a_1),x,y,z) {\,}{\,}{\,}{\,}{\,}{\,}
 \forall  (x,y,z) \\ \notag &{\,}\\ \notag & \in
{\,}\mathcal{M}(A)\cup \mathcal{M}(B)\cup
\mathcal{M}(D)\backslash\{\varphi_{CA}\varphi_{AC}(a_1),\varphi_{CB}\varphi_{AC}(a_1),\varphi_{CD}\varphi_{AC}(a_1)\}
\\ \notag & \mathcal{U}_A(\varphi_{CA}\varphi_{AC}(a_1),\varphi_{AC}\varphi_{CA}\varphi_{AC}(a_1),\varphi_{AB}\varphi_{CA}\varphi_{AC}(a_1),\varphi_{AD}\varphi_{CA}\varphi_{AC}(a_1))>
\\ \notag & \mathcal{U}_A(\varphi_{CA}\varphi_{AC}(a_1),x,y,z) {\,}{\,}{\,}{\,}{\,}{\,}
\\ \notag &\forall  (x,y,z) \\ \notag &{\,}\in {\,}\mathcal{M}(C)\cup \mathcal{M}(B)\cup \mathcal{M}(D)\backslash\{\varphi_{DA}\varphi_{CD}\varphi_{AC}(a_1),\varphi_{DB}\varphi_{CD}\varphi_{AC}(a_1),\varphi_{DC}\varphi_{CD}\varphi_{AC}(a_1)\}
\\ \notag & \mathcal{U}_D(\varphi_{AD}\varphi_{CA}\varphi_{AC}(a_1),\varphi_{DA}\varphi_{AD}\varphi_{CA}\varphi_{AC}(a_1),\varphi_{DA}\varphi_{AD}\varphi_{CA}\varphi_{AC}(a_1),\varphi_{DB}\varphi_{AD}\varphi_{CA}\varphi_{AC}(a_1))>\\ \notag & \mathcal{U}_D(\varphi_{AD}\varphi_{CA}\varphi_{AC}(a_1),x,y,z) {\,}{\,}{\,}{\,}{\,}{\,}
\\ \notag &\forall  (x,y,z) \\ \notag &{\,}\in {\,}\mathcal{M}(C)\cup \mathcal{M}(B)\cup \mathcal{M}(A)\backslash\{\varphi_{DA}\varphi_{AD}\varphi_{CA}\varphi_{AC}(a_1),\varphi_{DA}\varphi_{AD}\varphi_{CA}\varphi_{AC}(a_1),\varphi_{DB}\varphi_{AD}\varphi_{CA}\varphi_{AC}(a_1)\}
\end{align}

\textbf{Type 4}

\begin{align}
& \mathcal{U}_A(a_1,\varphi_{AC}(a_1),\varphi_{AD}(a_1),\varphi_{AB}(a_1))> \mathcal{U}_A(a_1,x,y,z) {\,}{\,}{\,}{\,}{\,}{\,}
 \\ \notag & \forall  (x,y,z) \\ \notag &{\,}\in {\,}\mathcal{M}(B)\cup \mathcal{M}(C)\cup \mathcal{M}(D)\backslash\{,\varphi_{AD}(a_1),\varphi_{AC}(a_1),\varphi_{AB}(a_1)\}
\\ \notag &
\mathcal{U}_B(\varphi_{AB}(a_1),\varphi_{BA}\varphi_{AB}(a_1),\varphi_{BC}\varphi_{AB}(a_1),\varphi_{BD}\varphi_{AB}(a_1))>
\mathcal{U}_B(\varphi_{AB}(a_1),x,y,z) {\,}{\,}{\,}{\,}{\,}{\,} \\
\notag &\forall (x,y,z)\\ \notag &{\,} \in {\,}\mathcal{M}(A)\cup
\mathcal{M}(C)\cup
\mathcal{M}(D)\backslash\{\varphi_{BA}\varphi_{AB}(a_1),\varphi_{BC}\varphi_{AB}(a_1),\varphi_{BD}\varphi_{AB}(a_1)\}
\\ \notag & \mathcal{U}_A(\varphi_{BA}\varphi_{AB}(a_1),\varphi_{AC}\varphi_{BA}\varphi_{AB}(a_1),\varphi_{AB}\varphi_{BA}\varphi_{AB}(a_1),\varphi_{AD}\varphi_{BA}\varphi_{AB}(a_1))>
\\ \notag & \mathcal{U}_A(\varphi_{BA}\varphi_{AB}(a_1),x,y,z) {\,}{\,}{\,}{\,}{\,}{\,}\forall  (x,y,z) \\ \notag &{\,}
\in {\,}\mathcal{M}(C)\cup \mathcal{M}(B)\cup
\mathcal{M}(D)\backslash\{\varphi_{AC}\varphi_{BA}\varphi_{AB}(a_1),\varphi_{AB}\varphi_{BA}\varphi_{AB}(a_1),\varphi_{AD}\varphi_{BA}\varphi_{AB}(a_1)\}
\\ \notag & \mathcal{U}_C(\varphi_{AC}\varphi_{BA}\varphi_{AB}(a_1),\varphi_{CA}\varphi_{AC}\varphi_{BA}\varphi_{AB}(a_1),\varphi_{CD}\varphi_{AC}\varphi_{BA}\varphi_{AB}(a_1),
\\ \notag &\varphi_{CB}\varphi_{AC}\varphi_{BA}\varphi_{AB}(a_1),\varphi_{DC}\varphi_{AD}\varphi_{CA}\varphi_{AC}(a_1))> \mathcal{U}_C(\varphi_{AD}\varphi_{CA}\varphi_{AC}(a_1),x,y,z) {\,}{\,}{\,}{\,}{\,}\\ \notag &{\,}\forall  (x,y,z) \\ \notag &{\,}\in {\,}\mathcal{M}(D)
\\ \notag & \cup \mathcal{M}(B)\cup \mathcal{M}(A)\backslash\{\varphi_{DA}\varphi_{AD}\varphi_{CA}\varphi_{AC}(a_1),\varphi_{DA}\varphi_{AD}\varphi_{CA}\varphi_{AC}(a_1),
\\ \notag &\varphi_{DB}\varphi_{AD}\varphi_{CA}\varphi_{AC}(a_1),\varphi_{DC}\varphi_{AD}\varphi_{CA}\varphi_{AC}(a_1)\}
\\ \notag & \mathcal{U}_D(\varphi_{CD}\varphi_{AC}\varphi_{BA}\varphi_{AB}(a_1),\varphi_{DA}\varphi_{CD}\varphi_{AC}\varphi_{BA}\varphi_{AB}(a_1),
\\ \notag &\varphi_{DB}\varphi_{CD}\varphi_{AC}\varphi_{BA}\varphi_{AB}(a_1),\varphi_{DB}\varphi_{CD}\varphi_{AC}\varphi_{BA}\varphi_{AB}(a_1),\varphi_{DC}\varphi_{CD}\varphi_{AC}\varphi_{BA}\varphi_{AB}(a_1))>
\\ \notag &\mathcal{U}_D(\varphi_{CD}\varphi_{AC}\varphi_{BA}\varphi_{AB}(a_1),x,y,z) {\,}{\,}{\,}{\,}{\,}{\,}\forall  (x,y,z) \\ \notag & \in \mathcal{M}(C)\cup \mathcal{M}(B)\cup \mathcal{M}(A)\backslash\{\varphi_{DA}\varphi_{CD}\varphi_{AC}\varphi_{BA}\varphi_{AB}(a_1),\\ \notag & \varphi_{DB}\varphi_{CD}\varphi_{AC}\varphi_{BA}\varphi_{AB}(a_1),\varphi_{DB}\varphi_{CD}\varphi_{AC}\varphi_{BA}\varphi_{AB}(a_1),\varphi_{DC}\varphi_{CD}\varphi_{AC}\varphi_{BA}\varphi_{AB}(a_1)\}
\end{align}
Player A would play $a_1$ if
players D, B and C play $d_2$, $b_2$ and $c_2$.
Following arrow C, player C would play $c_2$ if players A, B and D
play simultaneously $a_2$, $b_1$ and $d_2$. Following
arrow B at node 2, player B would play $b_1$ if players A, C and D
play $a_2$, $c_2$ and $d_1$. Following arrow A at
node 2, player A would play $a_2$ if players B, C and D play
$b_2$, $c_2$ and $d_2$. Following arrow D at node 4, player D would play $d_2$ if players A, B and C play $a_1$, $b_2$ and $c_1$.

We have reached the first periodic point. Following arrow
b at node 4, player b would play $b_2$ if players A, C and D play
$a_2$, $c_2$ and $d_2$. Following arrow C at node 4,
player C would play $c_2$ if players A, B and D play $a_2$, $b_1$
and $d_2$. Going back to node 2, following arrow D,
player D would play $d_2$ if players A, B and C play $a_1$, $b_2$
and $c_1$. Going back to node 1, following arrow B at
node 1, player B would play $b_2$ if players A, C and D play
$a_2$, $c_2$ and $d_1$. Following arrow A at node 4,
player A would play $a_2$ if players B, C and D play $b_2$, $c_2$
and $d_2$. Following arrow B at node 5, player B
would play $b_2$ if players A, C and D play $a_2$, $c_2$ and $d_2$. Player D would then play $d_2$, if players A, B and
C play $a_1$, $b_2$ and $c_1$. Following arrow D at
node 5, player D would play $d_2$ if players A, B and C play
$a_1$, $b_2$ and $c_1$. Following arrow C at node 5,
player C would play $c_2$ if players A, B and D play $a_2$, $b_1$
and $d_2$. Finally, $D$ would play $d_2$ if players
if players A B and C play $a_1$, $b_2$ and $c_1$.
\begin{figure}[h!]
\centering
\includegraphics[width=30pc]{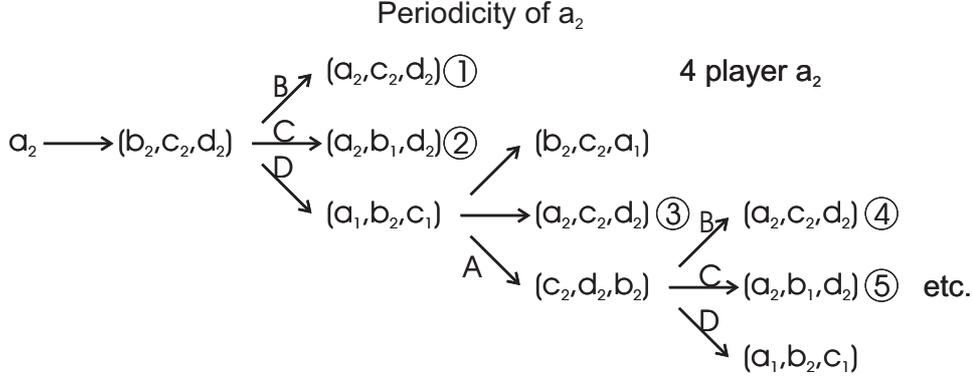}
\caption{Periodicity of strategy $a_1$ for the 4-player game of
Fig. \ref{4playergame}.}\label{4playera2}
\end{figure}

\subsection{The Periodicity Concept for  Simultaneous Perfect Information Multi-Player Games}

After these examples,  we shall now 
generalize the concept of periodicity to general multi-player simultaneous
perfect information strategic form games. Consider a finite player,
finite action, perfect information, simultaneous, strategic form
game, with
\begin{itemize}
\item The set  $I={1,\dots ,N}$ of  players
\item Their strategy spaces  $\mathcal{M}(i),\dots $ and the total strategy
space $\bar{G}=\mathcal{M}(1)\times ...  \mathcal{M}(N)$
\item Their payoff functions
$\mathcal{U}_i(\bar{G}):\bar{G}\rightarrow \Re$. 
\end{itemize}
We define $2N$ continuous maps between the strategy spaces $\mathcal{M}(i)$ and
$\mathcal{M}(j)$, 
\begin{equation}\label{phi1}
\varphi_{ij}:\mathcal{M}(i)\rightarrow \mathcal{M}(j)
\end{equation}
The maps act in such a way that, when starting with an action $x_i$ of player $i$, the following inequalities hold:
\begin{align}\label{inequalityies111}
&\mathcal{U}_i(x_i,\varphi_{ij}(x_i),\varphi_{ik}(x_i),...,\varphi_{il}(x_i))> \mathcal{U}_i(x_i,y_1,y_2,...,y_l) {\,}{\,}{\,}{\,}{\,}{\,}
 \\ \notag &\forall {\,}(y_1,y_2,...,y_l){\,}\in {\,}\mathcal{M}(j)\cup \mathcal{M}(k)\cup,...,\mathcal{M}(l)\backslash\{\varphi_{ij}(x_i),\varphi_{ik}(x_i),...,\varphi_{il}(x_i)\}
\\ \notag & \mathcal{U}_k(\varphi_{ik}(x_i),\varphi_{ki}\varphi_{ik}(x_i),\varphi_{kj}\varphi_{ik}(x_i),...,\varphi_{kl}\varphi_{ik}(x_i))> \mathcal{U}_k(\varphi_{ik}(x_i),y_1,y_2,...,y_l) {\,}{\,}{\,}{\,}{\,}{\,}
 \\ \notag &\forall {\,}(y_1,y_2,...,y_l){\,}\in {\,}\mathcal{M}(j)\cup \mathcal{M}(k)\cup,...,\mathcal{M}(l)\backslash\{\varphi_{ki}\varphi_{ik}(x_i),\varphi_{kj}\varphi_{ik}(x_i),...,\varphi_{kl}\varphi_{ik}(x_i)\}
\\ \notag &
\vdots
\\ \notag & n-step
\\ \notag & \mathcal{U}_m(\underbrace{\varphi_{mm_1}\varphi_{mm_1}\varphi_{m_1m_2}\circ...\circ \varphi_{ik}(x_i)}_{n{\,}times}...\underbrace{\varphi_{mk_1}\circ...\circ\varphi_{ik}(x_i)}_{n+1{\,}times})>
\\ \notag & \mathcal{U}_k(\varphi_{mm_1}\varphi_{mm_1}\varphi_{m_1m_2}\circ...\circ \varphi_{ik}(x_i),y_1,y_2,...,y_l) {\,}{\,}{\,}{\,}{\,}{\,}
 \\ \notag &\forall {\,}(y_1,y_2,...,y_l){\,}\in {\,}\mathcal{M}(j)\cup \mathcal{M}(k)\cup,...,\mathcal{M}(l)\backslash\{...\underbrace{\varphi_{mk_1}\circ...\circ\varphi_{ik}(x_i)}_{n+1{\,}times}\}
\end{align}
We call the action $x_i$ periodic if at some step of the
periodicity algorithm \cite{jostoikonomou}, we have
\begin{equation}\label{properdef}
x_i=\varphi_{mi}\circ...\circ\varphi_{ik}(x_i)
\end{equation}
Let us explain the meaning of each step of the algorithm. Start
with the first step, when $i$ plays $x_i$, his payoff is
maximized when his opponents play a combination of actions
(simultaneously), namely the actions
$(\varphi_{i1}(x_i),\varphi_{i2}(x_i),...,\varphi_{iN}(x_i))$. This procedure is repeated at every  step.

\begin{defi}[Periodicity]

In an $N$-player simultaneous move strategic form game
with finite actions, we define periodic strategies for player A
to be the subset  $\mathcal{P}(A)$ of his available strategies $\mathcal{M}(A)$ for which  there
exists an operator $\mathcal{Q}$: $\mathcal{M}(A)\rightarrow \mathcal{M}(A)$, with $\mathcal{Q}=\varphi_{ij}(x_i),\varphi_{ik}(x_i),...,\varphi_{il}(x_i)$ for which $\mathcal{Q}x_i=x_i$ such 
that the inequalities of relation (\ref{inequalityies111}) are fulfilled at each step.
\end{defi}
Periodic strategies are structures inherent to every non-trivial
finite action $N$-player strategic form game. 
\begin{theo}\label{1}
Every finite action simultaneous $N$-player strategic form game contains at least one periodic action.
\end{theo}

\begin{proof}
The proof of the theorem is very easy, since the inequalities
(\ref{inequalityies111}) hold. Let us consider player $i$ and  start from
an action $x_*$ which is assumed to be non-periodic. If we apply  the maps $\varphi{(ij)}$
to $x_*$, so that the inequalities
(\ref{inequalityies111}) are satisfied at every step, then, since
the game contains a finite number $n$ of actions, there
will be an action $x_a$ for which there exists an operator
constructed from a finite number of maps
$\mathcal{Q}=\underbrace{\varphi_{ij}(x_i),\varphi_{ik}(x_i),...,\varphi_{il}(x_i)}_{finite{\,}{\,}times}$,
so that $\mathcal{Q} x_a=x_a$. If the above is not true for any
other action apart from $x_a$, then since the game contains a
finite number of actions, this would imply that $x_a$ is periodic.
So every finite action game contains at least one periodic action.\\
A more  detailed
proof goes as follows. Suppose we start with the non-periodic action $x_i$ of player $i$.  Then
\begin{align}\label{fff}
&\mathcal{U}_i(x_i,\varphi_{ij}(x_i),\varphi_{ik}(x_i),...,\varphi_{il}(x_i))> \mathcal{U}_i(x_i,y_1,y_2,...,y_l)... {\,}{\,}{\,}{\,}{\,}{\,}
\end{align}
The algorithm will continue for some player $k$,
\begin{align}\label{mmm}
\mathcal{U}_k(\varphi_{ik}(x_i),\varphi_{ki}\varphi_{ik}(x_i),\varphi_{kj}\varphi_{ik}(x_i),...,\varphi_{kl}\varphi_{ik}(x_i))> \mathcal{U}_k(\varphi_{ik}(x_i),y_1,y_2,...,y_l).....
\end{align}
After this step, the algorithm will continue for some of the actions $\varphi_{ki}\varphi_{ik}\circ...\circ \varphi_{kl}\varphi_{ik}$, if none of the actions is repeated. Suppose the algorithm continues and it is the turn of player $m$, with 
\begin{align}\label{mmm}
\mathcal{U}_m(\varphi_{km}(x_i)\varphi_{ik}(x_i),\varphi_{mk}\varphi_{km}\varphi_{ik}(x_i),...,\varphi_{ml}\varphi_{km}\varphi_{ik}(x_i))> \mathcal{U}_m(\varphi_{km}(x_i)\varphi_{ik}(x_i),y_1,y_2,...,y_l).....
\end{align}
Since this is deterministic and there are only finitely many players and actions, it eventually has to become periodic. 
\end{proof}
The above reasoning  reveals
another property of the set of periodic actions in finite
multi-player simultaneous strategic form games. Recall
the definition of set stable strategies from Bernheim
\cite{bernheim}. We modify this definition of set stability as
follows:

\begin{defi} [Set Stability]

Let $\mathcal{Q}$ be an automorphism
$\mathcal{Q}:\mathcal{M}(A)\rightarrow \mathcal{M}(A)$. In
addition, let $A\subseteq A\cup B \subseteq \mathcal{M}(A)$, with
$A\cap B= \varnothing$. The set $A$ is set stable under the action
of the map $\mathcal{Q}$ if, for any initial $x_0$ $\in$ $A\cup B$
and any sequence $x_k$ formed by taking $x_{k+1}$ $\in$
$\mathcal{Q}(x_k)$, there exists $x_K$ $\in$ $A\cup B$ such that
$d(x_K,x^1)<\epsilon$, with $x^1$ $\in$ $A$. For finite sets, this
implies that any sequence formed by applying  the operator
$\mathcal{Q}$  produces an $x_k$ for any initial $x_0$,
with  $x_k$ belonging to the set stable set $A$.
\end{defi}

\begin{theo}\label{2}

Let $\mathcal{P}(i)$ denote the set of periodic strategies for player i. The set $\mathcal{P}(i)$ is set stable, under the action of the maps $\varphi_{ij}$.

\end{theo}

\noindent Thus,  the periodicity
diagram of any non-periodic action $x_0$ results in the
periodicity cycle of some action $x_K$.

\begin{proof} The proof of this theorem is contained in the proof of Theorem 1.
\end{proof}

\subsection{New Features; Remarks}

There is  one  difference between  the 2- and the
multi-player periodicity. In the two-player case, the utility
functions chain is 
\begin{align}\label{g2perc}
& U_A{\,}{\,}{\,}\xrightarrow{P}{\,}{\,}{\,} U_B{\,}{\,}{\,}\xrightarrow{P}{\,}{\,}{\,} U_A {\,}{\,}{\,}\xrightarrow{P}{\,}{\,}{\,} U_B....
\end{align}
and the periodicity occurs for $U_A$, if  we start with
a periodic action of player A. In the multi-player case,
although we may start with an action $x_i$ of player $i$ and
 the utility $U_i$, the periodicity might occur at the
utility function of another player, say $U_m$. Let us further
explain this version of periodicity. At the end of the
algorithm, player $m$  will play one of his actions, when his
opponents play some actions, one of which is $x_i$, corresponding
to player $i$. However, this does not exclude the fact that we
might return to the utility function of player $i$ again. One
example of this kind is Type 2 periodicity of player C, for the
three player game we studied previously in this section, or the
periodicity of $a_2$ corresponding to the same game. Having
studied the perfect information case, we now generalize our
framework to include non-perfect information games.

\section{Non Perfect Information Games -- Bayesian Games}

In this section we address the issue of periodicity in the case of
finite games with incomplete information. Our analysis on
incomplete information games is based mainly on references
\cite{rat3,harsanyi1,harsanyi2,harsanyi3,tan88,new2,new3,new4,new5,new6,new7,new8}
and references therein. In Bayesian strategic form games and more
generally in strategic form games with incomplete information, we
can always associate some relaated complete
information strategic form games  to the game in question. The corresponding
strategic form games are called ex-ante and interim strategic form
games. Exploiting these two, we will define and study the ex-ante
and interim rationalizable strategies and through these, the
periodicity in the case of non-perfect information games. With
respect to the latter, the interim rationalizability has two
versions, the interim independent and interim correlated
rationalizability. Both can be found by constructing the interim
independent and interim correlated strategic form game from the
initial Bayesian game. Since the Bayesian games can be represented
in terms of strategic form games, all the periodicity concepts
that we developed in the 2-player \cite{jostoikonomou} and
multi-player cases hold true. For simplicity, we shall only present the case with two
players and two actions for each player. The findings can be easily be generalized to the
multi-player case. The interim independent
strategic form game with the Bayesian game having initially two
players corresponds to a three player game. Let us start with the
ex-ante game.
A Bayesian game is a list $(N,A,\Theta,T,u,p)$, with

\begin{itemize}
 \item $N$, the number of players

\item $A=(A_i)_{i\in {\,}N}$, the set of action profiles with generic member $a=(a_i)_{i\in {\,}N}$

\item $\Theta$, the set of all possible parameters $\theta_i$ (in our case usually two different matrices for one of the two players)

\item $T=(T_i)_{i\in{\,}N}$ the set of types with generic member $t=(t_i)_{i\in{\,}N}$

\item $u_i:\Theta \times A\rightarrow R$, the payoff function of player i

\item $p_i=p_i(\cdot \mid t_i)$ $\in$ $\Delta (\Theta\times T_{-i})$ is the belief of the type $t_i$ about $\theta,t_{-i}$
\end{itemize}
Each player $i$ knows his own type $t_i$ but does not necessarily
know $\theta$, or the other players' types, about which he has a
belief $p_i(\cdot \mid t_i)$. The game is defined in terms of
players interim beliefs $p_i(\cdot \mid t_i)$, which they obtain
after they observe their own type, but before taking their action.
The game can also be defined by ex-ante beliefs $p_i$ $\in$
$\Delta (\Theta \times T)$ for some belief $p_i$. The game has a
common prior, if there exists $\pi$ $\in$ $\Delta (\Theta \times
T)$ such that:
\begin{equation}\label{cp}
p_i(\cdot \mid t_i)=\pi (\cdot \mid t_i ),{\,}\forall{\,}t_i{\,}\in{\,}T_i,{\,}\forall{\,}i{\,}\in{\,}N
\end{equation}
In that case, the game is denoted by $(N,A,\Theta,u,\pi)$. When
modelling incomplete information, there is often no ex-ante stage
or an explicit information structure in which players observe
values of some signals. In the modelling stage, each player $i$
has the following hierarchical belief system:

\begin{itemize}
 \item Some belief $\tau_i^1$ $\in$ $\Delta (\Theta)$, about the payoffs (and the other aspects of the physical world), a belief that is often referred to as the first order belief of $i$

 \item Some belief $\tau_i^2$ $\in$ $\Delta (\Theta\times \Delta (\Theta_{\Theta}))$ about the payoffs and the other players' first order beliefs $((\theta,\tau_{-i}^1))$

 \item Iteratively, for each $n$, some belief $\tau_i^n$ about the payoffs and the other players'  beliefs of all orders $<n$, $((\theta,\tau_{-i}^1,\tau_{-i}^2,\dots \tau_{-i}^{n-1}))$
\end{itemize}
In the Harsanyi type space formalism
\cite{harsanyi1,harsanyi2,harsanyi3}, the infinite belief
hierarchies are modelled using a type space $(\Theta ,T, P)$ and
also using a type $t_i$ $\in$ $T_i$ in the following way: Given a
type $t_i$ and a type space $(\Theta ,T, P)$, one can compute the
first order belief of a type $t_i$, by
\begin{equation}\label{h1}
h_i^1(\cdot \mid t_i)=\mathrm{marg}_{\theta} p(\cdot \mid t_i)
\end{equation}
so that
\begin{equation}\label{h2}
 h_i^1(\theta \mid t_i)=\sum_{t_{-i}}\mathrm{marg}_{\theta} p(\theta,t_{-i} \mid t_i)
\end{equation}
and the second order by
\begin{equation}\label{h3}
h_i^2(\theta,\hat{h}_{-i}^1)=\sum_{t_{-i}\mid h_{-i}^1(\cdot\mid t_{-i}=\hat{h}_{-i}^1}p(\theta, t_{-i}\mid t_i)
\end{equation}
A type space $(\Theta,t,p)$, and a type $t_i$ $\in$ $T_i$ model a belief hierarchy $(\tau_i^1,\tau_i^2,...)$ if
\begin{equation}\label{}
h_i^k(\cdot \mid t_i=\tau_i^k),{\,}\forall {\,}k
\end{equation}
Given any Bayesian game $(N,A,\Theta,u,\pi)$, with common prior
$\pi$, one can define the ex-ante game, which we denote by 
$G_{ex}=(N,S,U)$, where $S_i=A_i^{T_i}$ and
\begin{equation}\label{bt}
U_i=E_{\pi}[u_i(\theta,s(t)]
\end{equation}
for each $i$ $\in$ $N$ and $s$ $\in$ ${\,}$ $S$. For any Bayesian
game $(N,A,\Theta,T,u,p)$ one can also define the interim game,
which we denote by $G_{int}=(\hat{N},\hat{S},\hat{U})$, where
$\hat{N}=\cup_{i{\,}\in{\,}T_i}$ and also $\hat{S}_{t_i}=A_i$ for
each $t_i$ $\in$ $\hat{N}$ and
\begin{equation}\label{intg}
U_{t_i}(\hat{s})=E[u_i(\theta,\hat{s}_{t_{-i}}\mid p_i (\cdot \mid t_i)]=\sum_{(\theta,t_{-i})}u_{i}(\theta,\hat{s}_{t_{-i}})p(\theta,t_{-i}\mid t_i)
\end{equation}
for each $i$ $\in$ $N$ and $s$ $\in$ ${\,}$ $S$.

\subsection{Ex-ante game and Ex-ante Rationalizability}

Given any Bayesian game $(N,A,\Theta,T,u,p)$ and a player
$i \in N$, a strategy $s_i:T_i\rightarrow A_i$ is said to be
ex-ante rationalizable iff $s_i$ is rationalizable in the
corresponding ex-ante strategic form game $G_{ant}$
\cite{rat2,rat3,pereabook}. Ex-ante rationalizability makes sense
if there is an ex-ante stage in the game. In that case, ex-ante
rationalizability captures precisely the implications of common
knowledge of rationality as perceived in the ex-ante planning
stage of the game \cite{rat2,rat3}. It does impose unnecessary
restrictions on players' beliefs from an interim perspective
however. Let us look at the following example
\cite{rat2,rat3,pereabook}: Consider a Bayesian game with the
following characteristics:
\begin{itemize}
 \item $N=(1,2)$)
\item $\Theta = (\theta,\theta ')$
\item $T = (t_1,t_1')\times t_2$
\item $p(\theta, t_1,t_2)=p(\theta ', t_1',t_2)=\frac{1}{2}$
\end{itemize}

\newpage
The action space and the payoff functions are given by
\begin{table}[h]
\centering
  \begin{tabular}{| l |l |l | }
    \hline
      $\theta $ &  $L$ & $R$ \\ \hline
  $U$ & 1,$\epsilon$ & -2,0   \\ \hline
    $D$ & 0,0 & 0,1\\
    \hline
 \end{tabular},{\,}{\,}{\,}{\,} \begin{tabular}{| l |l |l | }
    \hline
      $\theta '$ &  $L$ & $R$ \\ \hline
  $U$ & -2,$\epsilon$ & 1,0   \\ \hline
    $D$ & 0,0 & 5,1\\
    \hline
 \end{tabular}
\end{table}

Here, player A has two types corresponding to two different
payoff actions. Player B has only one payoff table and one type.
The ex-ante representation of this game is equal to
\begin{table}[h]
\centering
  \begin{tabular}{| l |l |l | }
    \hline
       &  $L$ & $R$ \\ \hline
  $UU$ & $-1/2$,$\epsilon$ & $-1/2$,$\epsilon$   \\ \hline
    $UD$ & $1/2$,$\epsilon$ & $-1$,$1/2$\\
    \hline
$DU$ & $-1$,$\epsilon/2$ & $1/2$,$1/2$   \\ \hline
$DD$ & $0$,$0$ & 0,1   \\ \hline
 \end{tabular}
\end{table}

To every Bayesian game corresponds an ex-ante perfect information
strategic form game. The actions that are rationalizable in the
ex-ante strategic form game are called ex-ante rationalizable
actions. The rationalizable strategy profile in the case at hand
is $S^{\infty}(G_{ant}=(DU,R))$. The periodicity cycle of this
strategy is
\begin{align}\label{ccfgbfg1uhy23}
DU{\,}{\,}{\,}\xrightarrow{P}{\,}{\,}{\,} R{\,}{\,}{\,}\xrightarrow{P}{\,}{\,}{\,} DU
\end{align}
In addition, we can see that the theorem which relates types to
periodicity number holds true, since there are two  types needed to describes
this periodic cycles. In this case, the
types are the ones that correspond to the perfect information
ex-ante strategic form game, so these are seen in a perfect
information perspective. Of course all the theorems holding true
for finite simultaneous strategic form games, hold also true for
Bayesian games since the latter are equivalent to perfect
information strategic form games. We now proceed to interim
rationalizability related periodic equilibria.

\subsection{Interim Rationalizability}

There are conflicting  notions of interim
rationalizability in incomplete information games in the literature. One
straightforward notion of interim rationalizability is to apply
rationalizability to the interim game $G_{int}$. An embedded
assumption of the interim game is that it is common knowledge
that the belief of a player $i$ about $\theta_{-i}$, which is
given by $p_i(\cdot \mid t_i)$, is independent of his belief about
the other players' actions. In particular, his belief about
$(\theta,t_{-i},a_i)$ is derived from some belief $p_i(\cdot \mid
t_i)\times \mu_{t_i}$ for some $\mu_{t_i}$ $\in$ $\Delta
(A_{-i}^{T_{-i}})$. This is because we have taken the expectations
with respect to $p_i(\cdot \mid t_i)$, in defining the interim
game $G_{int}$, before considering his beliefs about the other
players' actions. Because of this independence assumption, such a 
rationalizability notion is called interim independent
rationalizability. Through the interim rationalizability we will
make contact with the periodicity concept in this case as well.

\subsubsection{Interim Independent Rationalizability}

Given any Bayesian game $B=(N,A,\Theta , T, u,p)$ and any type
$t_i$ of player $i$ $\in$ $N$, an action $a_i$ $\in$ $A_i$ is said
to be interim independent rationalizable for $t_i$, iff $a_i$ is
rationalizable for $t_i$ in the interim game $G_{int}$. The
interim independent Rationalizability is the most complex type of
rationalizability among all the rationalizability types for
Bayesian games. Consider the Bayesian game we used in the previous
example of the ex-ante game. The corresponding interim independent
game is actually a 3-player game with player-type set
$N=(t_1,t_1',t_2)$, and with the following payoff table:
\begin{table}[h]
\centering $U_1$
  \begin{tabular}{| l |l |l | }
    \hline
      $\theta $ &  $L$ & $R$ \\ \hline
  $U$ & 1,$\epsilon$,-2 & -2,0,1   \\ \hline
    $D$ & 0,$\epsilon /2$,-2 & 0,1/2,1\\
    \hline
 \end{tabular},{\,}{\,}{\,}{\,} $D_1$ \begin{tabular}{| l |l |l | }
    \hline
      $\theta '$ &  $L$ & $R$ \\ \hline
  $U$ & 1,$\epsilon /2$,0 & -2,1/2,0   \\ \hline
    $D$ & 0,0,0 & 0,1,0\\
    \hline
 \end{tabular}
\end{table}

The first player $t_1$ chooses the rows, the player $t_2$ the
columns and finally type $t_1'$ chooses the matrices. All actions
are rationalizable as can be easily checked. Let us see the
periodicity graphs for the above game.  For instance for $U$, 
the corresponding periodicity graph appears in Fig.
\ref{BAYESIANPERIOD}.

 This example is somewhat degenerate, but the periodicity study is identical to the study of periodicity in a 3-player strategic form game.
 This also proves that indirectly, using the interim rationalizability strategies, we relate the non-perfect information game to a multi-player,
 perfect information, simultaneous, strategic form game and therefore all the periodicity theorems hold true in this case as well .
 We further proceed in the same fashion and relate periodicity to the Interim Correlated Rationalizability concept.

\begin{figure}[h!]
\centering
\includegraphics[width=25pc]{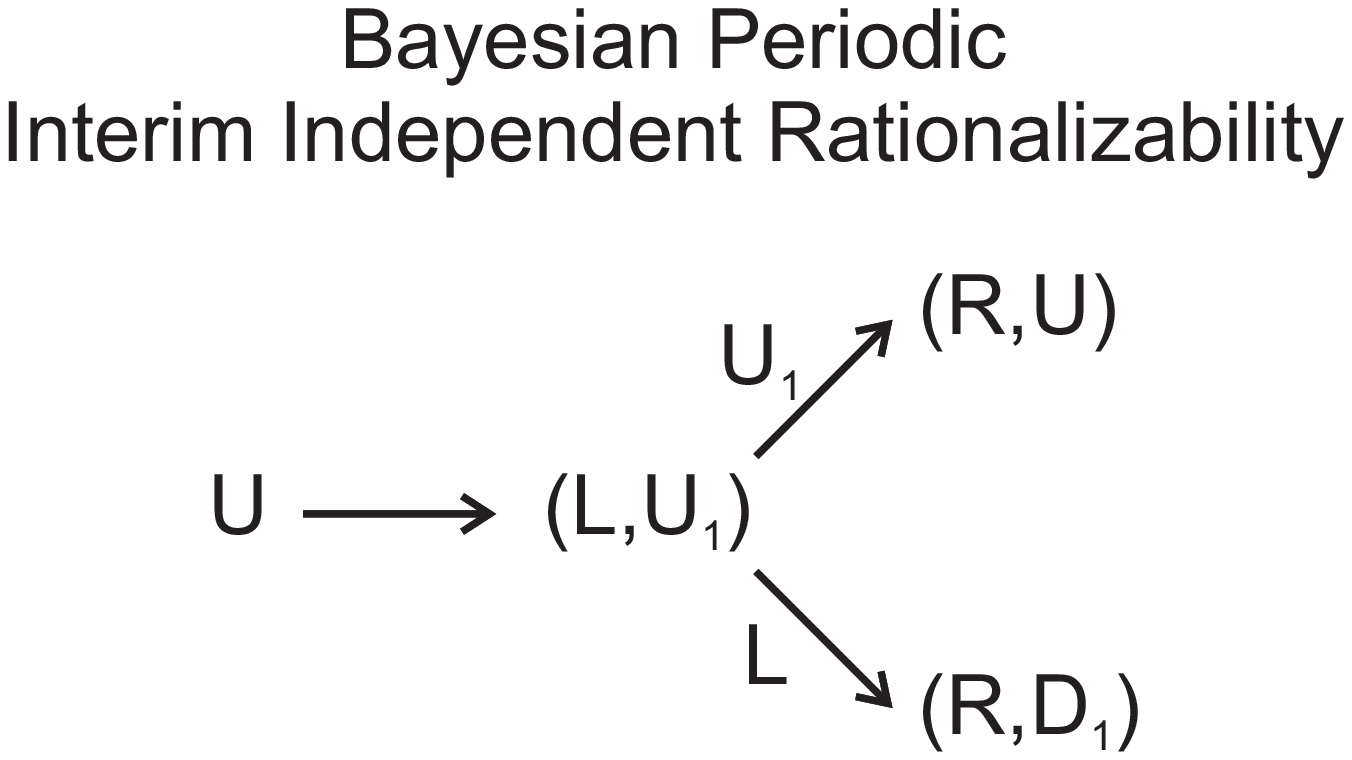}
\caption{Periodicity for a 3-player Bayesian Game. The first
player with type $t_1$ chooses the rows with actions $U$ and $D$,
the second player with type $t_2$ the columns, with actions $L$
and $R$ and finally the third player with type $t_1'$ chooses the
two matrices of the game.}\label{BAYESIANPERIOD}
\end{figure}

\subsubsection{Interim Correlated Rationalizability}

Consider a Bayesian game $B=(N,A,\Theta, T,u,p)$. Interim
correlated rationalizability \cite{rat2,rat3} allows more beliefs
than interim independent rationalizability, and it is a weaker concept
in reference to the latter. When all types have positive
probability, ex ante rationalizability is stronger than the other
two interim rationalizabilities. So all ex-ante rationalizable
actions are interim independent and all interim independent
rationalizable actions are interim correlated rationalizable
actions. The converse is not true. Thus the following holds true
\cite{rat2,rat3}:
\begin{equation}
\mathrm{ex-ante}\subset \mathrm{Interim-independent}\subset \mathrm{interim-correlated}
\end{equation}
Interim correlated rationalizability captures the implications of
common knowledge of rationality precisely \cite{rat2,rat3}. In
addition, interim independent rationalizability depends on the way
the hierarchies are modelled, in that there can be multiple
representations of the same hierarchy, with distinct sets of
interim independent rationalizable actions. Moreover, one cannot
have any extra robust prediction from refining interim correlated
rationalizability. Any prediction that does not follow from
interim correlated rationalizability alone relies on the
assumptions about the infinite hierarchy of beliefs. A researcher
cannot verify such a prediction in the modelling stage without the
knowledge of the  infinite hierarchy of beliefs. Now, the interim
correlated rationalizable actions are the ones that are
rationalizable in the interim correlated game. Let us see how this
game is found, by using a Bayesian game \cite{rat2,rat3}. Take
$\Theta =(-1,1)$, $N=(1,2)$ and the payoff matrices are:
\begin{table}[h]
\centering
  \begin{tabular}{| l |l |l | l |}
    \hline
      $\theta =1$ &  $b_1$ & $b_2$ & $b_3$\\ \hline
  $a_1$ & 1,1 & -10,10 & -10,0  \\ \hline
    $a_2$ & -10,-10 & 1,1& -10,0\\ \hline
$a_3$ & 0,-10 & 0,-10& 0,0 \\
    \hline \end{tabular},{\,}{\,}{\,}{\,}\begin{tabular}{| l |l |l | l |}
    \hline
      $\theta =-1$ &  $b_1$ & $b_2$ & $b_3$\\ \hline
  $a_1$ & -10,-10 & 1,1 & -10,0  \\ \hline
    $a_2$ & 1,1 & -10,-10& -10,0\\
    \hline
$a_3$ & 0,-10 & 0,-10& 0,0 \\
    \hline
 \end{tabular}
\caption{Game 1B}
\label{game1}
\end{table}

We consider the type space $T=(t_1,t_2)$, with $p(\theta =1,t)=p(\theta =-1,t)=1/2$. The interim game is the following complete information game:
\begin{table}[h]
\centering
  \begin{tabular}{| l |l |l | l |}
    \hline
      $\theta =1$ &  $b_1$ & $b_2$ & $b_3$\\ \hline
  $a_1$ & $-9/2$,$-9/2$ & $-9/2$,$-9/2$ & -10,0  \\ \hline
    $a_2$ & $-9/2$,$-9/2$ & $-9/2$,$-9/2$ & -10,0\\ \hline
$a_3$ & 0,-10 & 0,-10& 0,0 \\
    \hline \end{tabular}
\caption{Game 1B}
\label{game1}
\end{table}

It is easy to show that even in this Bayesian framework we can
find a periodic action and specifically in the interim reduced
game. Thereby, we indirectly demonstrated that by using the
various imperfect information rationalizability concepts, we
relate periodicity with Bayesian games in general. Therefore we
may formalize the periodicity concept in Bayesian games.

\subsubsection{Periodicity and Bayesian Games}

We can easily understand that since every Bayesian game
corresponds to some perfect information, finite player, finite
action, strategic form game, the following theorem holds.
\begin{theo}\label{1}
Every finite action simultaneous $N$-player Bayesian strategic form game contains at least one periodic action.
\end{theo}

\begin{proof}
Every finite player
finite action strategic form game corresponds to an interim game
or an ex-ante game, which are finite action finite player
games. Therefore since every finite action, finite player
strategic form game has at least a periodic action, it follows
that this is also true for every finite action, finite player, Bayesian strategic form
game.
\end{proof}

Moreover, all the arguments that hold  for perfect information
games also hold  for the ex-ante and interim representations of a
strategic form game. So we can generalize these arguments to
Bayesian games. For the ex-ante and interim correlated
representations of a Bayesian game, the following theorem holds.

\begin{theo}

In a two player perfect information ex-ante and interim correlated
representation of a two-player Bayesian strategic form game, the
number of types $N_{t_i}$ corresponding to the periodic cycle of
an ex-ante or interim correlated rationalizable periodic action
is
\begin{equation}\label{mainsecresults}
N_{t_i}=2{\,}n
\end{equation}
The types are those corresponding to the perfect information
representation of the Bayesian game and not those corresponding
to the incomplete information game.
\end{theo}
\begin{proof}

We shall call rationalizable strategies those which are
rationalizable for the corresponding ex-ante or interim correlated
strategic form game, without specifying to which we refer
\cite{rat2,rat3}. The results hold  for either case.
Having this in mind, for every such action, if the periodicity
number is $n$, it is possible to construct a periodic chain
with exactly $2n$ rationalizable actions appearing in that chain.
Therefore, we need to prove that for each action appearing in the
rationalizability chain there exists at least one type, so the
minimum number of types corresponding to all the actions of the
rationalizability chain is $2n$. As is proved in \cite{pereabook},
in a static game with finitely many choices for every player, it
is always possible to construct an epistemic model in which,
 \begin{itemize}
 \item Every type expresses common belief in rationality

\item Every type assigns for every opponent probability 1 to one specific choice and one specific type for that opponent.
\end{itemize}
Thus, for two player games, each type for player A, for example,
assigns probability 1 to one of his opponent's actions and one
specific type for that action, such that this action is optimal
for his opponent. In addition, in two player games, rationalizable
actions and choices that can be made under common belief in
rationality coincide. Hence, we can associate
to every rationalizable action of player A exactly one type which
in turn assigns probability 1 to one specific rationalizable
action and one specific type of his opponents type's and actions.
Moreover, as proved in \cite{pereabook}, the actions
that can rationally be made under common belief in rationality are
rationalizable. To state this more formally, in a static game with
finitely many actions for every player, the choices that can
rationally be made under common belief in rationality, are exactly
those choices that survive iterated elimination of strictly
dominated strategies. Hence, for two player games, we conclude
that strategies which express common belief in rationality and
rationalizable strategies coincide. This is
 because all beliefs in two-player games are independent. (This is not always true in games with more than two
 players, however.) Therefore, when periodic rationalizable strategies are
 considered, the total number of types needed for a rationalizability
 cycle is equal to $2n$. This concludes the proof.
\end{proof}

\section{Periodicity and
  Cooperativity}

While our concept of a periodic solution seems to involve some form of cooperativity, this is of course different from what is called cooperative game theory. The latter is about binding commitments, coalitions and the distribution of payoffs inside such coalitions. All these features are absent in our setting. For further illustration, we shall now discuss  one of the most refined cooperative game theory
concepts, that of a  cooperative-competitive (CO-CO) solution
 \cite{kalai} (see also \cite{bressan}) and we shall compare the results of this
solution concept with those that result from the periodic
strategies algorithm.

\subsection{Cooperative-Competitive Equilibrium}

Consider a general, two player non-zero sum game with players $A$
and $B$, described by the payoff functions $\Phi^A$ and $\Phi^B$,
with:
\begin{equation}\label{coco1}
\Phi^A:\mathcal{M}(A)\times \mathcal{M}(B)\rightarrow \Re,{\,}{\,}{\,}\Phi^b:\mathcal{M}(A)\times \mathcal{M}(B)\rightarrow \Re
\end{equation}
with the strategy spaces $\mathcal{M}(A)$ and $\mathcal{M}(B)$
being compact metric spaces, and the payoff functions being
continuous functions from $\mathcal{M}(A)\times \mathcal{M}(B)$
into $\Re$. If cooperativity and communication between players is
allowed, the players $A$ and $B$ can adopt a set of strategies
$(a^{\sharp},b^{\sharp})$ that maximizes their combined payoffs,
\begin{equation}\label{combinedpay}
V^{\sharp}=\Phi^A(a^{\sharp},b^{\sharp})+\Phi^B(a^{\sharp},b^{\sharp})=\mathrm{max}_{a,b{\,}\in{\,}\mathcal{M}(A)\times \mathcal{M}(B)}\Big{[}\Phi^A(a,b)+\Phi^B(a,b)\Big{]}
\end{equation}
The choice of the strategy $(a^{\sharp},b^{\sharp})$ may favor
 one player more than the other. In such a case, the player that is
better off must provide some incentive to the other player, in
order that he complies with the strategy
$(a^{\sharp},b^{\sharp})$. This incentive is actually a side
payment. Splitting the total payoff, $V^{\sharp}$ into two equal
parts will not be acceptable, because this does not reflect the
relative strength of the players and their personal contributions
to their cooperativity outcomes \cite{bressan}. A more realistic
approach was introduced by \cite{kalai} which we shall now describe.
Define the following game:
\begin{equation}\label{payoffsplit}
\Phi^{\sharp}(a,b)=\frac{\Phi^A(a,b)+\Phi^B(a,b)}{2},{\,}{\,}{\,}\Phi^S(a,b)=\frac{\Phi^A(a,b)-\Phi^B(a,b)}{2}
\end{equation}
These relations actually imply that the original game is
split into two games, a purely cooperative one, with payoff
$\Phi^{\sharp}(a,b)$, and a competitive one (which is a zero sum
game), with payoff $\Phi^S(a,b)$. In the cooperative game, the
players have  equal payoffs, that is, they both receive
$\Phi^{\sharp}(a,b)$, while in the purely competitive part, the
players have  opposite payoffs, namely $\Phi^S(a,b)$ and
$-\Phi^S(a,b)$.

\noindent Denote  the value of the zero-sum game by  $V^S$, with utility function $\Phi^S(a,b)$.

Having found the value of the game, the cooperative-competitive value of the game is defined as the  payoff pair
\begin{equation}\label{cc1}
(\frac{V^{\sharp}}{2}+V^S,\frac{V^{\sharp}}{2}-V^S)
\end{equation}
The cooperative-competitive solution of the game is defined as 
the pair of strategies $(a^{\sharp},b^{\sharp})$, together with a
side payment $\mathcal{P}_S$ from player B to player A, such that:
\begin{align}\label{ccsss}
&\Phi^A(a^{\sharp},b^{\sharp})+\mathcal{P}_S=\frac{V^{\sharp}}{2}+V^S
\\ \notag & \Phi^B(a^{\sharp},b^{\sharp})-\mathcal{P}_S=\frac{V^{\sharp}}{2}-V^S
\end{align}
Obviously, the side payment can be negative, in which case player A pays player B the amount $\mathcal{P}_S$.

Conceptually, the cooperative-competitive solution is opposite to the algorithm that yields periodic strategies, owing
to the fact that the cooperative-competitive solution, namely the
strategy pair $(a^{\sharp},b^{\sharp})$, is determined by
maximizing the sum of the player's and his opponent's utility. The
periodic strategies on the other hand are computed by maximizing
each player's own payoff, with respect to the opponent's actions.
We  shall now  present some characteristic examples and  compare
the cooperative-competitive solution and the periodic algorithm
solution.

\subsection{Cooperative-Competitive Solution and Periodicity Algorithm--Some Examples}

Consider the Battle of Sexes game that appears in Table \ref{bs}.
\begin{table}[h]
\centering
  \begin{tabular}{| l |l |l | }
    \hline
       &  $b_1$ & $b_2$ \\ \hline
  $a_1$ & 2,1 & 0,0   \\ \hline
    $a_2$ & 0,0 & 1,2\\
    \hline
 \end{tabular}
\caption{Battle of Sexes} \label{bs}
\end{table}

As we demonstrated in Ref. \cite{jostoikonomou}, for this game
both the pure strategy pairs $(a_1,b_1)$ and $(a_2,b_2)$ are
periodic strategies. Moreover, when we apply the periodic
strategies algorithm to mixed strategies, we obtain a mixed
strategy that yields the same payoffs as the mixed Nash
equilibrium, with the difference that each player's payoff does
not depend on his opponent's actions. Let us recall the results:

\noindent The mixed Nash equilibrium for this game is
$(p_N^*=\frac{2}{3},q_N^*=\frac{1}{3})$ and moreover, the
application of the periodic strategies algorithm yields the
strategy, $(p_p^*=1/3,q_p^*=2/3)$. The expected utilities of the
players are:
\begin{align}\label{periodicnashutilitiesforcomp}
&{\mathcal{U}_1}_{p,q}(p_p^*=1/3,q)=\frac{2}{3}, \\ \notag &
{\mathcal{U}_2}_{p,q}(p,q_p^*=2/3)=\frac{2}{3}, \\ \notag &
{\mathcal{U}_1}_{p,q}(p,q_N^*=1/3)=\frac{2}{3}, \\ \notag &
{\mathcal{U}_1}_{p,q}(p_N^*=2/3,q)=\frac{2}{3} \\ \notag &
\end{align}
Hence, the payoff corresponding to the mixed Nash equilibrium is
$({\mathcal{U}_1}_N,{\mathcal{U}_2}_N)=(2/3,2/3)$ and the
algorithm of periodic strategies yields the payoffs
$({\mathcal{U}_1}_P,{\mathcal{U}_2}_P)=(2/3,2/3)$. Let us now turn
 to the cooperative-competitive solution of the Battle of
Sexes game. By the procedure described in the previous
subsection,  the zero-sum game of the Battle of Sexes
game is given in table \ref{bs1}.
\begin{table}[h]
\centering
  \begin{tabular}{| l |l |l | }
    \hline
       &  $b_1$ & $b_2$ \\ \hline
  $a_1$ & 1/2 & {\,}{\,}{\,}0   \\ \hline
    $a_2$ & {\,}{\,}0 & -1/2\\
    \hline
 \end{tabular}
\caption{Battle of Sexes} \label{bs1}
\end{table}

We  compute  $V^{\sharp}=3$ and $V^S=0$. It is obvious that the
cooperative-competitive strategy is constituted from any of the
two strategy sets $(a_1,b_1)$ or $(a_2,b_2)$. Within the
cooperative-competitive solution, player B must make a side
payment $\mathcal{P}_S=1$ to player A. Hence, in the
cooperative-competitive solution the final utilities are
$({\mathcal{U}_1}_{CC},{\mathcal{U}_2}_{CC})=(2,2)$. As we can
see, when players cooperate, they receive a higher payoff than in all other non-cooperative payoffs we presented for
this game. Consequently, the strategies that are obtained from the
periodic strategies algorithm are, in expected utility terms, as
non-cooperative as the mixed Nash equilibrium.

\noindent Let us give another example of the non-cooperativity of the mixed and non-mixed periodic
strategies. Consider the game that appears in Table \ref{game1}.
\begin{table}[h]
\centering
  \begin{tabular}{| l |l |l | l |l |}
    \hline
       &  $b_1$ & $b_2$ & $b_3$& $b_4$ \\ \hline
  $a_1$ & 0,7 & 2,5 & 7,0& 0,1  \\ \hline
    $a_2$ & 5,2 & 7,7& 5,2& 0,1\\
    \hline
$a_3$ & 7,0 & 2,5& 0,7& 0,1\\
    \hline
$a_4$ & 0,0 & 0,-2& 0,0& 10,-1\\
    \hline
 \end{tabular}
\caption{Game 1B} \label{game1}
\end{table}

The payoffs corresponding to the mixed Nash equilibrium
($p_N^*=\frac{5}{6},q_N^*=\frac{48}{49}$) and the ones
corresponding to the periodic strategies algorithm
($p_p^*=1/49,q_p^*=\frac{48}{49}$) are
\begin{align}\label{payoffsexam2}
&{\mathcal{U}_1}_{P}(p_p^*=1/49,q)=\frac{146}{49} \\ \notag &
{\mathcal{U}_2}_{P}(p,q_p^*=1/6)=\frac{35}{6} \\ \notag &
{\mathcal{U}_1}_{N}(p,q_N^*=48/49)=\frac{146}{49} \\ \notag &
{\mathcal{U}_1}_{N}(p_N^*=5/6,q)=\frac{35}{6} \\ \notag &
\end{align}
The strategy $(a_1,b_2)$ corresponds to the
cooperative-competitive strategy. The values $V^{\sharp}$ and
$V^S$ are equal to $V^{\sharp}=56$ and $V^S=-\frac{3}{2}$, and
hence the side payment of player A to player B is
$\mathcal{P}_S=-\frac{47}{2}$. The cooperative-competitive value
of the game (the final payoffs of the two players) is
$({\mathcal{U}_1}_{CC},{\mathcal{U}_2}_{CC})=(\frac{53}{2},\frac{59}{2})$.
By comparing the cooperative payoffs with the non-cooperative
ones, appearing in equation (\ref{payoffsexam2}), it is obvious
that the non-cooperative ones are  smaller than
the cooperative ones. Thus,   the strategies that result from
applying the periodic strategies algorithm are again non-cooperative.

Nevertheless, for some games, the cooperative-competitive
strategies payoff value (in the terminology of
cooperative-competitive equilibria) may coincide with the periodic
mixed or pure strategies payoff. But this occurs only for a rather particular class of games, like the
Prisoner-Dilemma. For example, for the game  in Table \ref{prisdil},
\begin{table}[h]
\centering\label{prisdil}
  \begin{tabular}{| l |l |l | }
    \hline
       &  $B_1$ & $B_2$ \\ \hline
  $A_1$ & 4,4 & -1,6  \\ \hline
    $A_2$ & 6,-1 & 0,0\\
    \hline
 \end{tabular}
\caption{Prisoners Dilemma}
\end{table}

the application of the periodic strategies
results to the strategy pair $(A_1,B_1)$, with payoffs
$({\mathcal{U}_1}_{P},{\mathcal{U}_2}_{P})=(4,4)$. For this game
the values $V^{\sharp}$ and $V^S$ are equal to $V^{\sharp}=8$ and
$V^S=0$, and the side payment of player A to player B is
$\mathcal{P}_S=0$. Consequently, the cooperative-competitive value
of the game is
$({\mathcal{U}_1}_{CC},{\mathcal{U}_2}_{CC})=(4,4)$, which is the
same as the periodic one. However, this is accidental and  an
artifact of the details of the payoff matrix.

\section{Epistemic Game Theory Framework and Periodic Strategies}

In this section, we shall connect the periodicity number $n$
appearing in the automorphism $\mathcal{Q}^n$  defined earlier
to the number of types needed to describe a two player
simultaneous strategic form game within  an epistemic framework.
We shall  assume a perfect information context. The
epistemic game theory formalism was introduced  by  Harsanyi, in order to describe incomplete information
games \cite{harsanyi1,harsanyi2,harsanyi3} and thereafter adopted
by other authors (see for example \cite{rat1,rat2,rat3} and
references therein). Our approach mimics the one used in
\cite{tan88} and also the one adopted from Perea in
\cite{pereabook}. For completeness,  we
shall briefly present the appropriate formalism and reasoning.

\subsection{Belief Hierarchies in Complete Information Games and Types and Common Belief in Rationality}

Consider a two player game with a set of finite actions 
for each player, A and  B. A belief hierarchy for  player A of the game is
constructed from a chain of increasing order beliefs in terms of
objective probabilities as follows \cite{pereabook}:

\begin{itemize}
 \item A first order belief is the belief that player A holds for player B's actions

\item Iteratively, a $k-$th order belief represents the
belief that player A holds for the $(k-1)$-th order belief of
player B.

\end{itemize}
The belief hierarchy expresses in general rational
choices of the players under common
belief in rationality, that is, every player believes in his
opponent's rationality and believes that his opponent believes that
he acts rationally and so on. Since belief hierarchies are
 not so easy to  use in practice,  the concept of a type is introduced, which encompasses all the
information that a belief hierarchy contains, but is a more
compact way to describe such a hierarchy.

Before doing that, let us quantify the belief hierarchies in a
more formal way, in terms of  spaces of probability distributions. With a suitable topology and metric, the space of probability distributions on a compact metric space is again a compact metric space, and therefore, the construction can be iterated, that is, we can consider probability distributions on spaces of probability distributions.

The first
order belief hierarchy is given by all the
probabilities distributions over the  space of actions that 
player $i$ considers possible  for his opponents. By assumption, this set  $X_i^1$ is finite, hence in particular compact,  and we may also equip it with a metric. The space of
first order beliefs then is the space of probability distributions on that space, 
\begin{equation}\label{firblf}
B_i^1=\Delta (X_i^1)
\end{equation}
Iteratively, we obtain the
$k$-th order of uncertainty,
\begin{equation}\label{korblf}
X_i^k=X_i^{k-1}\times (\times_{j\neq i}B_j^{k-1})
\end{equation}
which embodies the $(k-1)$-th order space of uncertainty and also
the $(k-1)$-th order of the opponent's beliefs. Thus, the space
of $k$-th order beliefs is the set $\Delta (X_i^k)$. A belief
hierarchy $b_i$ for the player $i$  is an infinite chain of
beliefs $b_i^k$ $\in$ $B_i^k$, $\forall$ $k$, that is:
\begin{equation}\label{totbelhier}
b_i=(b_i^1,b_i^2,...,b_i^k)
\end{equation}
Relation (\ref{totbelhier}) encodes what was said  above. The belief hierarchy is assumed to be
coherent, which means that the various beliefs  in
the belief hierarchy do not contradict each other, that is, for
$m>k$
\begin{equation}\label{}
\mathrm{mrg}(b_i^m\lvert X_i^{k-1})=b_i^{k-1}
\end{equation}
Having defined coherent belief hierarchies, the epistemic
framework is constructed using the definition of an epistemic type
which is simply a coherent belief hierarchy for a player $i$. A
type corresponds to some epistemic model constructed for the game,
so let $T_i$ be the total number of types needed to describe
player $i$. In addition, for every player $i$ and for every
$t_i$ $\in$ $T_i$, the epistemic model specifies a probability
distribution $b_i(t_i)$ over the set $C_{-i}\times T_{-i}$, which
represents the set of choice-types of player $i$'s opponent $-i$.
The probability distribution $b_i(t_i)$ stands for the belief that
a player $i$'s type $t_i$ holds about player's $-i$ actions and
types, so
\begin{equation}\label{}
b_i:T_i\rightarrow \Delta (T_{-i}\times C_{-i})
\end{equation}
for a two player game. The type of a player $i$
is the  complete belief hierarchy. Now a choice $c_i$ of player $i$ is optimal for
his type $t_i$ if it is optimal for the first order
beliefs that $t_i$ holds about the opponent's choices. Within the
epistemic game theoretic framework, one can easily
define common belief in rationality. Indeed, we say that the type
$t_i$ believes in the opponent's rationality if $t_i$ assigns
positive probability to his opponents $-i$ choice types
$(c_{-i},t_{-i})$, in which case $c_{-i}$ is optimal for type
$t_{-i}$. Having defined the belief in opponent's rationality, we
define the $k-$fold belief in rationality \cite{pereabook}:

\begin{itemize}
 \item Type $t_i$ expresses 1-fold belief in rationality if $t_i$ believes in the opponent's rationality


\item Iteratively, type $t_i$ expresses $k$-fold belief in rationality if $t_i$ assigns positive probability to opponent types that express $(k-1)$-fold belief in rationality.

\item Type $t_i$ corresponding to player $i$ expresses common belief in rationality, if it expresses $k-$fold belief in rationality for every $k$.
\end{itemize}

\noindent In addition, we can formally define a {\it rational choice}, when common belief in rationality is assumed in the game, as follows: A choice $c_i$ of player $i$  is rational under common belief in rationality, if there is some type $t_i$ such that:
\begin{itemize}
\item Type $t_i$, expresses common belief in rationality
\item Choice $c_i$ is optimal for this type $t_i$
\end{itemize}
Our aim is to connect the periodicity number $n$ 
defined earlier to the number of types that are necessary to
describe a simultaneous two player finite action game. This
connection will use the point rationalizable
strategies.

\subsection{The Connection of the Periodicity Number to the total Number of Types of the Epistemic Model}

As  demonstrated in Ref. \cite{jostoikonomou} the
rationalizable actions that are also periodic are particularly
interesting, since for these we can connect the total periodicity
number $n$ to the numbers of types needed to describe the
game with an epistemic model. This relation can be described by
the following theorem:

\begin{theo}\label{4}

In a two player perfect information strategic form game, the number
of types $N_{t_i}$ corresponding to the periodic cycle of a
rationalizable periodic action is
\begin{equation}\label{mainsecresults}
N_{t_i}=2{\,}n
\end{equation}
\end{theo}
\begin{proof}

For every such action if the periodicity number is $n$, it is
possible to construct a periodicity chain with exactly $2n$
rationalizable actions appearing in that chain. Therefore what is
necessary to prove is that for each action appearing in the
rationalizability chain, there exist at least one type, so the
minimum number of types corresponding to all the actions of the
rationalizability chain is $2n$. As  proved in \cite{pereabook},
in a static game with finitely many choices for every player, it
is always possible to construct an epistemic model in which,
 \begin{itemize}
 \item Every type expresses common belief in rationality

\item Every type assigns for every opponent probability 1 to one
specific choice and one specific type for that opponent.
\end{itemize}
Therefore, for two player games, each type for player A for example,
assigns probability 1 to one of his opponents actions and one
specific type for that action, such that this action is optimal
for his opponent. In addition, in two player games, rationalizable
actions and choices that can be made under common belief in
rationality coincide. Hence, we can associate
to every rationalizable action of player A exactly one type which
in turn assigns probability 1 to one specific rationalizable
action and one specific type of his opponent's types and actions.
Moreover, as  proved in \cite{pereabook}, the actions that can
rationally be made under common belief in rationality are
rationalizable. To state this more formally, in a static game with
finitely many actions for every player, the choices that can
rationally be made under common belief in rationality are exactly
those  that survive iterated elimination of strictly
dominated strategies. Hence, for two player games, we conclude
that strategies which express common
 belief in rationality and rationalizable strategies coincide. This is
 because all beliefs in two-player games are independent, something
 that is not always true in games with more than two
 players. Therefore, when periodic rationalizable strategies are
 considered, the total number of types needed for a rationalizability
 cycle is equal to $2n$.
\end{proof}

\subsubsection{A Comment on Simple Belief Hierarchies and Nash Equilibria}

Within an epistemic game theory context, a type $t_i$ is said to
have a simple belief hierarchy, if $t_i$'s belief hierarchy is
generated by some combination $\sigma_i$ of probabilistic beliefs
about the players choices. Thus, a type has a simple
belief hierarchy if it is believed that his opponents are correct
about his beliefs. As  proved in
\cite{pereabook}, a simple belief hierarchy, given by  probabilistic beliefs $\sigma_i$ about players' choices,
expresses common belief in rationality, if the combination
$\sigma_i$ of beliefs is itself a Nash equilibrium. The converse
is not always true. Hence, using the theorem above, the number of
types needed to describe a simple belief hierarchy for a Nash
equilibrium is 2. Obviously, if a Nash action is periodic, then
$n=1$ and applying relation (\ref{mainsecresults}), we find that
the types needed in the periodic Nash case are two.

\noindent There is  an interesting point regarding simple
belief hierarchies. When considering two player games, it is
proved (see \cite{pereabook}, theorem 4.4.3) that a type $t_i$ has
a simple belief hierarchy iff $t_i$ believes that his opponent
holds correct beliefs and believes that his opponent believes that
he holds correct beliefs himself. Thus, he believes that he does not err in his prediction about his opponent's beliefs, and he
believes that for his opponent too. In higher order beliefs this
is no longer true, and therefore we could argue that the total
number of wrong beliefs of all the two players about each other's
beliefs is equal to $2n-1$. Thus, the total number of
errors of the two players is  $2n-1$. Errors here are  the beliefs $\sigma_i$ due to which the higher order
belief hierarchy fails to be a simple belief hierarchy.

\section*{Concluding Remarks}

In this work we have studied extensions and generalizations of  the
periodicity concept  introduced in
\cite{jostoikonomou}. In particular, we have shown  the
existence of periodic strategies in multi-player perfect
information simultaneous strategic form games. We also proved that the set of
periodic strategies is set-stable under the periodicity map. In
addition, we discussed the presence of periodic strategies in
games with incomplete information, focusing on Bayesian games. In
that case we made extensive use of various generalizations of
Bernheim's rationalizability concept. The issue of cooperativity
and periodicity was formally addressed as well. The
periodic strategies are simply as cooperative as the mixed Nash
equilibrium. In an epistemic framework, the number of types
needed to describe the rationalizability cycle of a rationalizable
periodic strategy  equals twice the periodicity number of
that action. The next step  would be the
inclusion of mixed strategies in multi-player games. Actually, the
cooperativity issue in games with more than two players becomes
more complex, because  the players are free to form coalitions. Periodicity then has
to be reconsidered under this perspective.

Clearly, the periodicity feature for finitely many actions of
strategic form games can be very useful. Indeed, all the periodic
actions can be found using some simple program. This result is actually a common
feature of every non-degenerate finite action game, that is, every
non-Nash rationalizable action is usually periodic. This can be
very useful for games that have, as we mentioned, finitely many
actions, since the potential non-Nash rationalizable actions can
be determined by finding the periodic strategies. Furthermore, an
interesting future study would be to consider 3-player mixed
strategies and their relation to periodic strategies. One should
carefully examine whether there is any exceptional class of games
with the special attributes of the two player games that we
presented in the present article. In particular, we should check
whether the algorithm of periodic strategies leads to strategies
for which the expected utility of players is higher
than the corresponding Nash one, and in addition if the periodic
strategies for a player are independent of the other player's
action,  as in the two player case. In addition, the
multi-player cooperativity issue should also be formally
addressed. The question whether the periodic strategies imply any
sort of cooperativity has to be re-addressed in a multi-player
context. This is because, in cases with $N\geq 3$ players, two or
more players may form coalitions in order to cooperate against
the rest. Moreover, one can investigate  the case of continuum utility functions. Finally, in the case
of Bayesian games, one might look for a  connection between the types of
the imperfect information case and the corresponding Ex-ante or
interim game, or  a connection between  periodicity  imperfect information types spaces.

An important feature of  periodic
strategies as examined in this paper is that they  make a player robust against the way that the
opponent-rival decides to play the game. In contrast to the Nash
strategies, where each player relies on his opponent's rationality
and on the fact that the opponent will actually play the Nash
strategy too,  the  payoff of a player that uses a periodic strategy is not affected by the opponent's actual actions. This is valuable
in non-trivial games, like the  prisoner's dilemma. It is remarkable that although we used a
non-trivial non-cooperative context, we ended up that the optimal
equilibrium of the game is the socially optimal solution. In this
work we demonstrated how periodic strategies can be realized in  
multi-player simultaneous perfect information games and also in
games with imperfect information. Hence this shows that the
periodicity concept seems to be an inherent feature of every
non-trivial game. The advantage  of the periodic strategies over
the Nash strategies is  that the periodic
strategies players do not depend on the rationality of the
opponent. Although rationality is considered a prerequisite in
most  games,  there exist many modern politics and
economics related examples where rationality is questioned. More
importantly, in many cases the opponents may have hidden
information, so although a player might think that the payoff are
given and the game is played with perfect information about the
payoffs of the game, the opponent might act non-rationally with
respect to the perfect information game, but rationally with
respect to the hidden information game. 
the periodic strategies then are safe strategies in the sense that  the possibility of loosing  is minimized or
controlled in a formal way.

\end{document}